\documentclass[11p,reqno]{amsart}

\usepackage[T1]{fontenc}
\usepackage{graphicx}
\usepackage{amssymb,amsthm,amsmath,mathrsfs,bm,braket,marginnote}
\usepackage{enumerate}
\usepackage{appendix}
\usepackage[colorlinks=true, pdfstartview=FitV, linkcolor=blue, citecolor=blue, urlcolor=blue]{hyperref}
\usepackage{multirow}

\usepackage{pgf}
\usepackage{pgfplots}
\usepackage{tikz}
\usetikzlibrary{arrows,calc}
\usepackage{verbatim}
\usetikzlibrary{decorations.pathreplacing,decorations.pathmorphing}
\usepackage[numbers,sort&compress]{natbib}
\usepackage{dsfont}

\numberwithin{equation}{section}
\newtheorem{theorem}{Theorem}[section]

\newtheorem{definition}[theorem]{Definition}
\newtheorem{remark}[theorem]{Remark}

\newtheorem{proposition}[theorem]{Proposition}

\reversemarginpar
\newcommand{\mr}{\mathcal{R}}

\newcommand{\p}{\partial}

\newcommand{\z}{\mathbb{Z}}

\newcommand{\dl}{\{\!\!\{}
\newcommand{\dr}{\}\!\!\}}
\newcommand{\tr}{{\rm tr}}
\newcommand{\fig}{Fig.$\,$}

\newcommand{\wt}{{\rm wt}}
\newcommand{\A}{\mathcal{A}}

\author[B. Wang]{Bao Wang}
\address{School of Mathematics and Statistics, Ningbo University, Ningbo 315211, PR China.}
\email{wangbao@nbu.edu.cn}
\author[S.-H. Li]{Shi-Hao Li}
\address{ Department of Mathematics, Sichuan University, Chengdu, 610064, PR China.}
\email{shihao.li@scu.edu.cn}

\begin{document}

	\title{On non-commutative leapfrog map}

	\keywords{non-commutative leapfrog map;
		discrete integrable systems;
		non-commutative orthogonal polynomials;
		Poisson geometry.}
	\begin{abstract}

		We investigate the integrability of the non-commutative leapfrog map in this paper. 
		Firstly, we derive the explicit  formula for the non-commutative leapfrog map and corresponding discrete zero-curvature equation by employing the concept of non-commutative cross-ratio. 
		Then we revisit this discrete map, as well as its continuous limit, from the perspective of non-commutative Laurent bi-orthogonal polynomials. 
		Finally, the Poisson structure for this discrete non-commutative map is formulated with the help of a non-commutative network. 
		We aim to enhance our understanding of the integrability properties of the non-commutative leapfrog map and its related mathematical structures through these analysis and constructions.
		
	\end{abstract}
	
	\maketitle
	\section{Introduction}
	
	The pentagram map was first introduced by Schwartz \cite{1992The} as a discrete map on the polygons in a projective plane.
	It maps a polygon $P$ with vertices $\{v_i\}_{i\in\mathbb{Z}}$ to a new polygon $T(P)$ where the $i$th vertex is the intersection of diagonals $(v_i,v_{i+2})$ and $(v_{i+1},v_{i+3})$.
	Although the construction is simple,
	the resulting map exhibits many remarkable properties.
	These properties include but are not limited to the following intriguing characteristics:
	
	(i)
	Ovsienko, Schwartz and Tabachnikov proved that the pentagram map was Liouville integrable on the space of twisted polygons \cite{Ovsi2010}.
	Its Lax representation and corresponding algebraic-geometric integrability is established in \cite{Soloviev2013Integrability}.

	(ii) Glick demonstrated a profound connection between the pentagram map and Y-mutations of a specific cluster algebra \cite{Y2011}.
	Then Glick and Pylyavskyy generalized it to a rich family of pentagram-type maps \cite{Y2016}, which could be descibed as $Y$-mutations in a cluster algebra.
	
	(iii) Gekhtman, Shapiro, Tabachnikov and Vainshtein generalized Glick's construction to specify the pentagram map as a family of discrete integrable maps. The idea is to use the compatibility of Poisson and cluster structures and Poisson geometry of directed networks on surfaces \cite{gekhtman2016}.
	
	(iv) Ovenhouse generalized the approach of Gekhtman et al. to establish a non-commutative version of integrability of the non-commutative pentagram map \cite{Non}.
	The generalization encompasses the Grassmannian pentagram map proposed by Marí Beffa and Felipe \cite{grassmann}.
	
	(v) Generalizations of the pentagram map to polygons in higher dimensions have been extensively studied, such as \cite{Y2016,gekhtman2016,dented2016,long2022,higher2012,wang23}.
	The integrability of these generalized maps are also explored. 
	
	As a one-dimensional counterpart of the pentagram map, the leapfrog map introduced in  \cite{gekhtman2016} is defined as follows.
	Let $S=\{S_i\}_{i\in\mathbb{Z}}$ and $S^-=\{S_i^-\}_{i\in\mathbb{Z}}$ be a pair of $N$-gons in the projective line $\mathbb{RP}^1$.
	The leapfrog map $T$ is a map $(S^-,S)\mapsto (S,S^+)$,
	where the points in $S^+$ are defined by a sequence of projective transformations $\phi_i\in\mathrm{GL}_2(\mathbb{R})$ such that
	\begin{align*}
		\phi_i\left(S_{i-1},S_{i},S_{i+1},S_i^-\right)=\left(S_{i+1},S_i,S_{i-1},S_i^+\right).
	\end{align*}
	When points of polygons take values in an associative but not commutative algebra,
	we can get a non-commutative leapfrog map, whose
	complete integrability was listed as an unsolved problem in \cite{open2018}.
	Therefore, in this paper, 
	we mainly focus on the integrability of the non-commutative leapfrog map and its associated mathematical structures.
	
	Given that the non-commutative leapfrog map is a dynamical system in $\mathbb{P}^1$ over an associate algebra $\mathcal{R}$,
	it is necessary to consider a non-commutative version of projective geometry,
	such as the concept of non-commutative cross-ratios introduced by Retakh \cite{retakh14}.
	The non-commutative cross-ratios are defined by using a non-commutative analogue of the determinant, called the quasi-determinant, 
	as well as corresponding quasi-Pl\"ucker coordinates.
	These topics are thoroughly  reviewed in section \ref{sec_ncr} for the description of the non-commutative leapfrog map.
	
	In section \ref{sec_formula},
	we present an explicit formula for the non-commutative leapfrog map.
	Initially,
	Retakh introduced this map by utilizing the non-commutative cross-ratio \cite{retakh14,retakh20} as 
	\begin{align*}
		\kappa(S_{i-1},\,S_{i+1},\,S_{i},\,S_i^-)= \beta_i^{-1}\kappa(S_{i+1},\,S_{i-1},\,S_{i},\,S_i^+)\beta_i,
	\end{align*}
	where $\beta_i\in\mathcal{R}$ and $\kappa(x,y,z,w)$ represents the non-commutative cross-ratio of points $x,y,z,w\in\mathbb{P}^1$.
	Then we demonstrate that the non-commutative leapfrog map could be defined in three different ways.
	One is to construct a sequence of projective transformations $\phi_i\in GL_2(\mathcal{R})\times (\mathcal{R}^2)^{\otimes 4}$ such that $$\phi_i(S_{i-1},\,S_{i+1},\,S_{i},\,S_i^-)=(S_{i+1},\,S_{i-1},\,S_{i},\,S_i^+).$$
	The construction of projective transformation allows us to express the leapfrog map in terms of coordinates of points. 
	The second method is to consider a suitable scaling, which helps us to formulate the non-commutative leapfrog map by using cross-ratio type coordinates.
	The last formula is to make use of the $y$-variables
	\begin{align*}
		\begin{split}
			(b_i^j)^{-1}
			\left(1+y_{i-1}^{j}\right)
			\left(1+\left(y_i^j\right)^{-1}\right)^{-1}
			\left(y_i^{j-1}\right)^{-1}b_i^j\\
			=
			\left(1+\left(y_i^j\right)^{-1}\right)
			y_i^{j+1}
			\left(1+y_{i+1}^j\right)^{-1},
		\end{split}
	\end{align*}
	where $b_i^j$ could be expressed in terms of cross-ratios.
	This formula exhibits a connection to the $Y$-mutations in certain cluster algebra \cite{Y2016}.

	The dynamics of the leapfrog map leads to the so-called relativistic Toda equation in the theory of classical integrable systems,
	which exhibits a potential connection to the Laurent bi-orthogonal polynomials \cite{op2022}.
	Therefore, in section \ref{sec_ops},
	we construct a non-commutative version of Laurent bi-orthogonal polynomials.
	The use of discrete spectral transformations for non-commutative Laurent bi-orthogonal polynomials is made to realize the non-commutative leapfrog map.
	Additionally,
	we investigate a continuum limit of the discrete relativistic Toda equation by considering the continuous time evolution of moments, which gives an explanation for continuous evolution of the leapfrog map.
	This procedure provides us a deeper understanding of the non-commutative leapfrog map. On the one hand, the Lax integrability of this map is given in terms of non-commutative Laurent bi-orthogonal polynomials and corresponding spectral transformations. On the other hand, the solutions for the non-commutative leapfrog map are expressed by quasi-determinants, from which a clear description for the algebraic structure of this map is given.

	In section \ref{sec_Poisson},
	our focus shifts towards studying the Poisson structure of the non-commutative leapfrog map and demonstrating its integrability. 
	We first recall some basic facts about the double Poisson bracket on the representation space of an associate algebra $\mr$ by following the work of van den Bergh \cite{van2008double}. This 
	promotes us to give a double Poisson brackets over a non-commutative network, which is a non-commutative generalization of that given in \cite{gekhtman2016}. It is shown that the Postnikov move of the non-commutative network is in fact a discrete time evolution of the non-commutative leapfrog map, preserving the Poisson brackets for edge weights. The Lax representation of the non-commutative leapfrog map is re-expressed in terms of the boundary measurement matrix at the end.
	
	\section{non-commutative cross-ratios}\label{sec_ncr}
	
	In this section, we show explicit formulas for non-commutative cross-ratios by making use of quasi-Pl\"ucker coordinates. During this process, the expression of quasi-determinants is widely used to formulate explicit algebraic structures. Therefore, we give a brief introduction of quasi-determinants in the appendix for self-consistency. For more details about quasi-determinants and related topics, one could refer to \cite{gelfand91,gelfand97,gelfand05}.
	
	\subsection{Quasi-Pl\"ucker coordinates and their properties}
	Let $\mathcal{R}$ be a skew field and $\text{GL}_n(\mr)$ be all invertible $n\times n$ matrices over $\mr$. One could define the following quasi-Pl\"ucker coordinates \cite{retakh14}.
	\begin{definition}\label{def2.1}
		Given a $2\times n$ matrix 
		\begin{align*}
			A=\left(\begin{array}{cccc}
				a_{11}&a_{12}&\cdots&a_{1n}\\
				a_{21}&a_{22}&\cdots&a_{2n}
			\end{array}
			\right)\in\mr^{2\times n},
		\end{align*}
		if $i\ne j$, we call
		\begin{align*}
			q_{jk}^i(A)=\left|\begin{array}{cc}
				a_{1i}&\boxed{a_{1j}}\\
				a_{2i}&a_{2j}
			\end{array}
			\right|^{-1}\left|\begin{array}{cc}
				a_{1i}&\boxed{a_{1k}}\\
				a_{2i}&a_{2k}
			\end{array}
			\right|
		\end{align*}
		as the quasi-Pl\"ucker coordinate of $A$ over $\mr$.
	\end{definition}
	There are several properties regarding with quasi-Pl\"ucker coordinates. Firstly, it is noted that 
	a quasi-Pl\"ucker coordinate could be alternatively expressed by a single quasi-determinant
	\begin{align}\label{alt}
		q_{jk}^i(A)=-\left|\begin{array}{ccc}
			a_{1i}&a_{1k}&a_{1j}\\
			a_{2i}&a_{2k}&a_{2j}\\
			0&\boxed{0}&1
		\end{array}
		\right|,
	\end{align}
	which could be easily verified by making use of a non-commutative Jacobi identity \eqref{ncj1}. 
	Therefore, we sometimes refer a quasi-determinant in the form of \eqref{alt} as a quasi-Pl\"ucker coordinate.
	Moreover, quasi-Pl\"ucker coordinates have the following properties if actions on the coordinate matrix $A$ is considered.
	\begin{proposition}\label{prop3}
		If $g\in \text{GL}_2(\mr)$, then
		$
		q_{jk}^i (g\cdot A)=q_{jk}^i (A).
		$
	\end{proposition}
	\begin{proof}
		This proposition is easily verified if $g$ is an invertible upper/lower triangular matrix. 
		For an invertible upper triangular matrix, we have
		\begin{align}\label{triangular}
			\left|\left(\begin{array}{cc}
				g_{11}&0\\
				g_{21}&g_{22}
			\end{array}\right)\left(\begin{array}{cc}
				a_{1i}&a_{1j}\\
				a_{2i}&a_{2j}
			\end{array}\right)
			\right|_{22}
			=\left|\begin{array}{cc}
				g_{11}a_{1i}&g_{11}a_{1j}\\
				g_{21}a_{1i}+g_{22}a_{2i}&\boxed{g_{21}a_{1j}+g_{22}a_{2j}}
			\end{array}
			\right|
			=g_{22}\left|\begin{array}{cc}
				a_{1i}&a_{1j}\\
				a_{2i}&\boxed{a_{2j}}
			\end{array}\right|,
		\end{align}
		where $|\cdot|_{22}$ means the expansion of a quasideterminant from the $(2,2)$-position.
		Similarly, we have 
		\begin{align*}
			\left|\left(\begin{array}{cc}
				g_{11}&g_{12}\\
				0&g_{22}
			\end{array}
			\right)\left(\begin{array}{cc}
				a_{1i}&a_{1j}\\
				a_{2i}&a_{2j}
			\end{array}
			\right)
			\right|_{22}=g_{22}\left|\begin{array}{cc}
				a_{1i}&a_{1j}\\
				a_{2i}&\boxed{a_{2j}}
			\end{array}
			\right|.
		\end{align*}
		Therefore, when $g$ is a general invertible matrix, we have the following LU-decomposition
		\begin{align*}
			\left(\begin{array}{cc}
				g_{11}&g_{12}\\
				g_{21}&g_{22}
			\end{array}\right)=\left(\begin{array}{cc}
				1&0\\
				g_{21}g_{11}^{-1}&1
			\end{array}
			\right)\left(\begin{array}{cc}
				g_{11}&g_{12}\\
				0&g_{22}-g_{21}g_{11}^{-1}g_{12}
			\end{array}
			\right).
		\end{align*}
		By using \eqref{triangular}, we have
		\begin{align}\label{quasi1}
			\left|\left(\begin{array}{cc}
				g_{11}&g_{12}\\
				g_{21}&g_{22}
			\end{array}\right)\left(\begin{array}{cc}
				a_{1i}&a_{1j}\\
				a_{2i}&a_{2j}
			\end{array}\right)
			\right|_{22}=(g_{22}-g_{21}g_{11}^{-1}g_{12})\left|\begin{array}{cc}
				a_{1i}&a_{1j}\\
				a_{2i}&\boxed{a_{2j}}
			\end{array}
			\right|,
		\end{align}
		which gives the result by the definition of quasi-Pl\"ucker coordinates Definition \ref{def2.1}.
	\end{proof}

	\begin{proposition}\label{prop4}
		Let $\Lambda=\text{diag}(\lambda_1,\cdots,\lambda_n)\in\text{GL}_n(\mr)$. Then 
		\begin{align*}
			q_{jk}^i(A\cdot\Lambda)=\lambda_j^{-1}q_{jk}^{i}(A)\lambda_k.
		\end{align*}
	\end{proposition}
	\begin{proof}
		Noting that 
		\begin{align*}
			A\cdot\Lambda=\left(\begin{array}{ccc}
				a_{11}\lambda_1&\cdots&a_{1n}\lambda_n\\
				a_{21}\lambda_1&\cdots&a_{2n}\lambda_n
			\end{array}
			\right),
		\end{align*}
		we get the formula
		\begin{align*}
			q_{jk}^i (A\cdot\Lambda)&=\left|\begin{array}{cc}
				a_{1i}\lambda_i&a_{1j}\lambda_j\\
				a_{2i}\lambda_i&\boxed{a_{2j}\lambda_j}
			\end{array}
			\right|^{-1}\left|\begin{array}{cc}
				a_{1i}\lambda_i&a_{1k}\lambda_k\\
				a_{2i}\lambda_i&\boxed{a_{2k}\lambda_k}
			\end{array}
			\right|\\
			&=\left(
			\left|\begin{array}{cc}
				a_{1i}&a_{1j}\\
				a_{2i}&\boxed{a_{2i}}
			\end{array}
			\right|\lambda_j
			\right)^{-1}\left(
			\left|\begin{array}{cc}
				a_{1i}&a_{1k}\\
				a_{2i}&\boxed{a_{2k}}
			\end{array}
			\right|\lambda_{k}
			\right)=\lambda_j^{-1}q_{jk}^i(A)\lambda_k.
		\end{align*}
	\end{proof}
	The followings are some properties of quasi-Pl\"ucker coordinates when the coordinate matrix $A$ is fixed. We denote $q_{jk}^i$ instead of $q_{jk}^i(A)$ where it can't lead to a confusion. 
	\begin{proposition}
		If $j\ne i$ and $j\ne k$, then $(q_{ik}^j)^{-1}=q_{ki}^j$. 
	\end{proposition}
	This proposition could be directly verified from the definition of quasi-Pl\"ucker coordinates.
	\begin{proposition}
		The quasi-Pl\"ucker coordinates satisfy the following non-commutative skew symmetry relation 		\begin{align}\label{ncss}
			q_{ij}^kq_{jk}^iq_{ki}^j=-1.
		\end{align}
	\end{proposition}
	\begin{proof}
		According to the quasi-determinant expressions for quasi-Pl\"ucker coordinates \eqref{alt}, and make use of Non-commutative Jacobi identity \eqref{ncj1}, we have
		\begin{align*}
			q_{ij}^k=-\left|\begin{array}{ccc}
				a_{1k}&a_{1j}&a_{1i}\\
				a_{2k}&a_{2j}&a_{2i}\\
				0&\boxed{0}&1
			\end{array}
			\right|
			=\left|\begin{array}{cc}
				a_{1j}&a_{1i}\\
				\boxed{0}&1
			\end{array}
			\right|-\left|\begin{array}{cc}	
				a_{1k}&a_{1i}\\
				\boxed{0}&1
			\end{array}
			\right|\left|\begin{array}{cc}
				a_{1k}&a_{1i}\\
				\boxed{a_{2k}}&a_{2i}	
			\end{array}
			\right|^{-1}\left|\begin{array}{cc}
				a_{1j}&a_{1i}\\
				\boxed{a_{2j}}&a_{2i}
			\end{array}
			\right|,
		\end{align*} 
		from which we get the linear relation
		\begin{align*}
			q_{ij}^k=a_{1i}^{-1}(a_{1k}+a_{1j}q_{kj}^i).
		\end{align*}
		Therefore, by directly expanding the formula, we could obtain
		\begin{align*}
			q_{ij}^kq_{jk}^i
			&=-\left|\begin{array}{cc}
				a_{1i}&a_{1k}\\
				1&\boxed{0}
			\end{array}
			\right|+\left|\begin{array}{cc}
				a_{1i}&a_{1j}\\
				1&\boxed{0}
			\end{array}
			\right|\cdot\left|\begin{array}{cc}
				a_{1i}&a_{1j}\\
				a_{2i}&\boxed{a_{2j}}
			\end{array}
			\right|^{-1}\cdot\left|\begin{array}{cc}
				a_{1i}&a_{1k}\\
				a_{2i}&\boxed{a_{2k}}
			\end{array}
			\right|\\
			&=-\left|\begin{array}{ccc}
				a_{1i}&a_{1k}&a_{1j}\\
				a_{2i}&a_{2k}&a_{2j}\\
				1&\boxed{0}&0
			\end{array}
			\right|=-q_{ik}^j.
		\end{align*}
		Use of the relation $(q_{ik}^j)^{-1}=q_{ki}^j$ is made to complete the proof.
	\end{proof}
	Moreover, we have the following non-commutative Pl\"ucker identity.
	\begin{proposition}
		Quasi-Pl\"ucker coordinates satisfy the following non-commutative Pl\"ucker identity
		\begin{align*}
			q_{ij}^kq_{ji}^l+q_{il}^kq_{li}^j=1.
		\end{align*}
	\end{proposition}
	\begin{proof}
		This proposition could be proved by acting non-commutative Jacobi identity \eqref{ncj1} onto the quasi-determinant
		\begin{align*}
			A=\left|\begin{array}{cccc}	
				a_{1l}&a_{1i}&a_{1j}&a_{1k}\\
				a_{2l}&a_{2i}&a_{2j}&a_{2k}\\
				0&1&0&0\\
				\boxed{0}&0&1&0
			\end{array}
			\right|.
		\end{align*}
		On the one hand, if we apply non-commutative Jacobi identity to the $(1,4)$-columns and $(3,4)$-rows, we know that $A=q_{jl}^i-q_{jk}^i(q_{ik}^j)^{-1}q_{il}^j$. On the other hand, if we apply the identity to the $(1,2)$-columns and $(3,4)$-rows, we have $A=q_{jl}^k$. Therefore, we obtain the relation
		\begin{align}\label{ncplucker}
			q_{jl}^k=q_{jl}^i-q_{jk}^i(q_{ik}^j)^{-1}q_{il}^j=q_{jl}^i-q_{jk}^iq_{kl}^j.
		\end{align}
		By making use of \eqref{ncss}, we complete the proof.
	\end{proof}

	\subsection{Non-commutative cross-ratios and their properties}
	
	Let 
	\begin{align*}
		x=\left(\begin{array}{c}
			x_1\\x_2
		\end{array}
		\right),\quad y=\left(\begin{array}{c}
			y_1\\y_2
		\end{array}
		\right),\quad z=\left(\begin{array}{c}
			z_1\\z_2
		\end{array}
		\right),\quad t=\left(\begin{array}{c}
			t_1\\t_2
		\end{array}
		\right)
	\end{align*}
	represent four arbitrary vectors in $\mr^2$, and corresponding non-commutative cross-ratio $\kappa=\kappa(x,y,z,t)$ can be defined by equations
	\begin{align}\label{cr}
		\left\{\begin{array}{l}
			t=x\alpha+y\beta\\
			z=x\alpha\gamma+y\beta\gamma\cdot\kappa
		\end{array}
		\right.
	\end{align}
	where $\alpha,\beta,\gamma,\kappa\in\mr$. Using the coordinate matrix
	\begin{align*}
		X=\left(\begin{array}{cccc}
			x_1&y_1&z_1&t_1\\
			x_2&y_2&z_2&t_2
		\end{array}
		\right)\in\mr^{2\times 4},
	\end{align*}
	we show that non-commutative cross-ratio can be expressed by quasi-Pl\"ucker coordinates.
	\begin{proposition}
		The non-commutative cross-ratio defined by \eqref{cr} could be written in terms of quasi-Pl\"ucker coordinates as
		\begin{align}\label{nccr}
			\kappa(x,y,z,t)=q_{zt}^y q_{tz}^x.
		\end{align}
	\end{proposition}
	\begin{proof}
		According to the first equation in \eqref{cr}, we could write it into coordinates and
		\begin{align*}
			\left\{\begin{array}{l}
				t_1=x_1\alpha+y_1\beta,\\
				t_2=x_2\alpha+y_2\beta.
			\end{array}
			\right.
		\end{align*}
		Solving this linear system with non-commutative coefficients and Proposition \ref{p-ls}, we get
		\begin{align*}
			\alpha=(1,0)\left(\begin{array}{cc}
				x_1&y_1\\
				x_2&y_2
			\end{array}
			\right)^{-1}\left(\begin{array}{c}
				t_1\\t_2\end{array}
			\right)=-\left|\begin{array}{ccc}
				x_1&y_1&t_1\\
				x_2&y_2&t_2\\
				0&1&\boxed{0}
			\end{array}
			\right|=q_{xt}^y,
		\end{align*}
		and similarly $\beta=q_{yt}^x$. 
		By the second equation in \eqref{cr}, one gets
		\begin{align*}
			\gamma=(1,0)\left(\begin{array}{cc}
				x_1\alpha&y_1\beta\\
				x_2\alpha&y_2\beta
			\end{array}
			\right)^{-1}\left(\begin{array}{c}
				z_1\\z_2\end{array}
			\right)=(\alpha^{-1},0)\left(\begin{array}{cc}
				x_1&y_1\\
				x_2&y_2
			\end{array}
			\right)^{-1}\left(\begin{array}{c}
				z_1\\z_2
			\end{array}
			\right)=\alpha^{-1}q_{xz}^y.
		\end{align*}
		Moreover, according to the non-commutative skew symmetry relation \eqref{ncss}, the formula $\gamma=-q_{tz}^y$ is obtained. Therefore, the non-commutative cross-ratio could be computed via the formula
		\begin{align*}
			\kappa=\gamma^{-1}\beta^{-1}q_{yz}^x=-q_{zt}^yq_{ty}^xq_{yz}^x=q_{zt}^yq_{tz}^x.
		\end{align*}
	\end{proof}
	The followings demonstrate how is the non-commutative cross-ratio influenced by the actions on the coordinate matrix $X$. In fact, the coordinate space $\mr^{2\times 4}$ is a $(\text{M}_2(\mr),\text{M}_4(\mr))$-bimodule, which means that there exists a left multiplication action 
	\begin{align}\label{gl2}
		g\cdot X\mapsto gX,\quad g\in M_2(\mr), \,X\in\mr^{2\times 4},
	\end{align}
	and a right multiplication action 
	\begin{align*}
		X\cdot \Lambda\mapsto X\Lambda,\quad \Lambda\in M_4(\mr), \, X\in\mr^{2\times 4}.
	\end{align*}
	Especially, we are mainly interested in the case $g\in\text{GL}_2(\mr)$ and 
	\begin{align*}
		\Lambda\in(\mr^{\times})^{\otimes 4}=\{\text{diag}(\lambda_1,\lambda_2,\lambda_3,\lambda_4),\,\lambda_i\in\mr^\times\}
	\end{align*}
	where $\mr^\times$ denotes the group of all invertible elements in $\mr$. The following theorem demonstrates an explicit relation of non-commutative cross-ratio under the action $\text{GL}_2(\mr)\times(\mr^{\times})^{\otimes 4}$.
	
	\begin{theorem}\label{thm2}
		Let $X=(x,y,z,t)\in\mr^{2\times 4}$, $g\in \text{GL}_2(\mr)$ and $\Lambda=\text{diag}(\lambda_1,\lambda_2,\lambda_3,\lambda_4)\in(\mr^\times)^{\otimes 4}$, then
		\begin{align*}
			\kappa(g\cdot X \cdot \Lambda)	=\kappa(gx\lambda_1,gy\lambda_2,gz\lambda_3,gt\lambda_4)=\lambda_3^{-1}\kappa(x,y,z,t)\lambda_3.
		\end{align*} 
	\end{theorem}
	\begin{proof}
		This theorem could be easily verified by using Propositions \ref{prop3} and \ref{prop4}.
	\end{proof}
	In the commutative case, $\lambda_3^{-1}$ is eliminated by $\lambda_3$, and the cross-ratio is an invariant under the action of linear fractional transformations. 
	For the non-commutative cross-ratio defined above, it is a relative invariant under the action of $\text{GL}_2(\mr)\times (\mr^\times)^{\otimes 4}$ on the space $\mr^{2\times 4}$. Now we can claim the following theorem \cite[Thm. 8]{young82}, from which an integrable non-commutative leapfrog map is introduced.
	\begin{theorem}\label{thm_equi}
		There exists a transformation $\phi\in\text{GL}_2(\mr)\times (\mr^\times)^{\otimes 4}$
		\begin{align*}
			\phi:\quad\mr^{2\times 4}&\to \mr^{2\times 4}\\
			(x,y,z,t)&\mapsto(x',y',z',t'):=(gx\lambda_1,gy\lambda_2,gz\lambda_3,gt\lambda_4)
		\end{align*}
		such that 
		\begin{align}\label{cross_pro}
			\kappa(x,y,z,t)=\lambda_3\kappa(x',y',z',t')\lambda_3^{-1}
		\end{align}
		for some $\lambda_3\in\mr^\times$.
	\end{theorem}
	
	Besides, if a coordinate matrix $X$ is given, then non-commutative cross-ratios have the following relations if we exchange the positions of vectors $x,y,z,t$.
	\begin{align}\label{eq_3}
		\kappa(x,y,z,t)=\kappa(w,y,z,t)\kappa(x,w,z,t)
	\end{align}
	\begin{align}\label{eq_1-}
		\kappa(x,y,z,t)=1-\kappa(t,y,z,x)
	\end{align}
	\begin{align}\label{eq_per}
		q^x_{tz}\kappa(x,y,z,t)q^x_{zt}
		=q^y_{tz}\kappa(x,y,z,t)q^y_{zt}
		=\kappa(y,x,t,z)
	\end{align}

	\section{non-commutative leapfrog map}\label{sec_formula}
	Let $S=\{S_i\}_{i\in\mathbb{Z}}$ and $S^-=\{S_i^-\}_{i\in\mathbb{Z}}$ be two sequences of infinite points in the projective line $\mathbb{P}^1$ over $\mr$. Then there are actions of $\text{GL}_2(\mr)\times (\mr^\times)$ acting on points in $S$ and $S^{-}$, and we define the non-commutative leapfrog map by constructing a map on $(S^-,S)$.
	\begin{definition}\label{def_leap}
		The non-commutative leapfrog map $T$ is a map from $(S^-,S)$ to another sequences of points $(S,S^+)$ such that $T(S^{-},S)=(S,S^+)$, where points in $S^+$ are defined by transformations $\{\phi_i\in\text{GL}_2(\mr)\times (\mr^2)^{\otimes 4},i\in\mathbb{Z}\}$ such that $$\phi_i(S_{i-1},S_i,S_{i+1},S_i^{-})=(S_{i+1},S_i,S_{i-1},S_i^{+}).$$
	\end{definition}	
	According to Definition \ref{def_leap}, it is known that the points in $S^+$ are determined by the map as follows. Given three points $S_{i-1},S_i,S_{i+1}\in S$, then there exists a projective transformation $g_i\in\text{GL}_2(\mr)$ and $\alpha_i,\beta_i,\gamma_i\in\mr^\times$ such that
	\begin{align}\label{cond}
		(g_iS_{i-1}\alpha_i^{-1},\,g_iS_i\beta_i^{-1},\,g_iS_{i+1}\gamma_i^{-1})= (S_{i+1},\,S_i,\,S_{i-1}).
	\end{align}
	Moreover, if the projective transformation $g_i$ is uniquely determined by \eqref{cond}, then $S_i^+$ is defined by the formula
	\begin{align}\label{sp}
		S_i^+:=g_iS_i^-\eta_i^{-1}, \quad \eta_i\in\mr^\times.
	\end{align}
	According to Theorem \ref{thm_equi}, we know that $T$ is a non-commutative leapfrog map if and only if the corresponding non-commutative cross-ratios defined by $(S^-,S)$ and $(S,S^+)$ satisfy the relation
	\begin{align}\label{eq_k}
		\kappa(S_{i-1},\,S_{i+1},\,S_{i},\,S_i^-)= \beta_i^{-1}\kappa(S_{i+1},\,S_{i-1},\,S_{i},\,S_i^+)\beta_i,
	\end{align}
	where $\beta_i$ coincides with that in \eqref{cond}.
	Moreover, if we lift the vertices of $S, S^-$ and $S^+$ in $\mathbb{P}^1$ to vectors in $\mathcal{R}^2$, i.e. 
	\begin{align*}
		S_i^-= \begin{pmatrix}
			v_i^-\\
			1
		\end{pmatrix},\quad
		S_i= \begin{pmatrix}
			v_i\\
			1
		\end{pmatrix},\quad 
		S_i^+=\begin{pmatrix}
			v_i^+\\
			1
		\end{pmatrix},
	\end{align*}
	then under such homogeneous coordinates, the non-commutative cross-ratio relation \eqref{eq_k} becomes 
	\begin{align*}
		&(v_{i}-v_{i+1})^{-1}(v^{-}_i-v_{i+1})(v_i^--v_{i-1})^{-1}(v_{i}-v_{i-1})\\
		&\quad=\beta_i^{-1}(v_{i}-v_{i-1})^{-1}(v_i^+-v_{i-1})(v_i^+-v_{i+1})^{-1}(v_{i}-v_{i+1})\beta_i,
	\end{align*} 
	which has been given in \cite[Sec.7]{retakh20} as an equivalent definition of the non-commutative leapfrog map. 
	
	\subsection{Explicit formulas for the leapfrog map}
	Now we determine the projective transformation $\phi_i=(g_i,\Lambda_i)\in\text{GL}_2(\mr)\times (\mr^\times)^{\otimes 4}$ sending $(S_{i-1},S_i,S_{i+1},S_i^-)$ to $(S_{i+1},S_i,S_{i-1},S_i^+)$.
	\begin{theorem}
		If we take $\beta_i=-1$ in \eqref{cond}, then $g_i\in\text{GL}_2(\mr)$ has the form
		\begin{align*}
			g_i=\left(\begin{array}{cc}
				v_ir_i+1&-(v_ir_i+2)v_i\\
				r_i&-(r_iv_i+1)
			\end{array}
			\right),
		\end{align*}
		where $r_i$ is given by $r_i=(v_{i-1}-v_i)^{-1}+(v_{i+1}-v_i)^{-1}$. Moreover, $S_i^+$ is uniquely determined and has the coordinate
		\begin{align}\label{eq_v+}
			v_i^+=(v_{i-1}p_i+v_iq_i)(p_i+q_i)^{-1},
		\end{align}
		where $p_i=(v_{i-1}-v_i)^{-1}(v_{i+1}-v_i)$ and $q_i=(v_{i+1}-v_i^-)^{-1}(v_i-v_i^-)$.
	\end{theorem}
	
	\begin{proof}
		According to equation \eqref{cond}, it is known that $g_i$ satisfies the following equation
		\begin{align}\label{eq_g}
			g_i\left(
			\begin{pmatrix}
				v_{i-1}\\
				1
			\end{pmatrix},
			\begin{pmatrix}
				v_{i}\\
				1
			\end{pmatrix},
			\begin{pmatrix}
				v_{i+1}\\
				1
			\end{pmatrix}
			\right)
			=\left(
			\begin{pmatrix}
				v_{i+1}\\
				1
			\end{pmatrix}\alpha_i,
			\begin{pmatrix}
				v_{i}\\
				1
			\end{pmatrix}\beta_i,
			\begin{pmatrix}
				v_{i-1}\\
				1
			\end{pmatrix}\gamma_i
			\right),
		\end{align}
		where $\alpha_i$, $\beta_i$, $\gamma_i\in\mr^\times$.
		Taking the lifts of $S_{i-1},\,S_i$ in $\mr^2$, we find that $S_{i+1}$ could be written as a right linear combination of  $S_{i-1}$ and $S_i$, which is equivalent to
		\begin{align}\label{recur_v}
			\begin{pmatrix}
				v_{i+1}\\
				1
			\end{pmatrix}
			=\begin{pmatrix}
				v_{i-1}\\
				1
			\end{pmatrix}p_i
			+\begin{pmatrix}
				v_{i}\\
				1
			\end{pmatrix}(1-p_i),
		\end{align}
		where $p_i=(v_{i-1}-v_i)^{-1}(v_{i+1}-v_i)$.
		By substituting \eqref{recur_v} into \eqref{eq_g},
		we get
		\begin{align*}
			&g_i\begin{pmatrix}
				v_{i-1}\\
				1
			\end{pmatrix}
			=\begin{pmatrix}
				v_{i-1}\\
				1
			\end{pmatrix}p_i\alpha_i
			+\begin{pmatrix}
				v_{i}\\
				1
			\end{pmatrix}(1-p_i)\alpha_i,\\
			&g_i\begin{pmatrix}
				v_{i}\\
				1
			\end{pmatrix}
			=\begin{pmatrix}
				v_{i}\\
				1
			\end{pmatrix}\beta_i,\\
			&g_i\begin{pmatrix}
				v_{i-1}\\
				1
			\end{pmatrix}p_i
			+g_i\begin{pmatrix}
				v_{i}\\
				1
			\end{pmatrix}(1-p_i)
			=\begin{pmatrix}
				v_{i-1}\\
				1
			\end{pmatrix}\gamma_i,
		\end{align*}
		which implies
		\begin{align*}
			\begin{pmatrix}
				v_{i-1}\\
				1
			\end{pmatrix}(p_i\alpha_ip_i-\gamma_i)
			+\begin{pmatrix}
				v_{i}\\
				1
			\end{pmatrix}\left((1-p_i)\alpha_i p_i+\beta_i(1-p_i)\right)
			=0.
		\end{align*}
		Noting that $S_{i-1}$ and $S_i$ are different points and thus their homogeneous coordinates are linearly independent, we can get
		\begin{align*}
			\gamma_i=p_i\alpha_ip_i,\quad
			\alpha_i=-(1-p_i)^{-1}\beta_i(1-p_i)p_i^{-1}.
		\end{align*}
		Taking $\beta_i=-1$,
		we have $\alpha_i=p_i^{-1}$ and $\gamma_i=p_i$.
		In this case, by direct calculation, we have
		\begin{align*}
			g_i=\begin{pmatrix}
				v_ir_i+1&-(v_ir_i+2)v_i\\
				r_i&-(r_iv_i+1)
			\end{pmatrix},
		\end{align*}
		where $r_i=(v_{i-1}-v_i)^{-1}+(v_{i+1}-v_i)^{-1}$.
		Moreover, by using elementary row and column operations, we know that
		\begin{align*}
			g_i\rightarrow \left(\begin{array}{cc}
				v_ir_i+1&-v_i\\
				r_i&-1
			\end{array}
			\right)\rightarrow \left(\begin{array}{cc}
				1&0\\
				r_i&-1
			\end{array}
			\right),
		\end{align*}
		and thus $rank(g_i)=2$, which means $g_i\in\text{GL}_2(\mr)$.
		
		On the other hand, by writing $S_i^-$ as the right linear combination of $S_{i-1}$ and $S_i$, we have the following expression
		\begin{align}\label{eq_v-}
			\begin{pmatrix}
				v_{i}^-\\
				1
			\end{pmatrix}
			=\left(\begin{pmatrix}
				v_{i+1}\\
				1
			\end{pmatrix}
			-\begin{pmatrix}
				v_{i}\\
				1
			\end{pmatrix}q_i\right)
			(1-q_i)^{-1},
		\end{align}
		where $q_i=(v_{i+1}-v_i^-)^{-1}(v_i-v_i^-)$.
		Moreover, according to the determination equation \eqref{sp}, we know that 
		\begin{align*}
			\begin{pmatrix}
				v_{i}^+\\
				1
			\end{pmatrix}
			=g_i\begin{pmatrix}
				v_{i}^-\\
				1
			\end{pmatrix}\eta_i^{-1}&=\left(
			g_i\begin{pmatrix}
				v_{i+1}\\
				1
			\end{pmatrix}-g_i
			\begin{pmatrix}
				v_{i}\\
				1
			\end{pmatrix}q_i
			\right)
			(1-q_i)^{-1}\eta_i^{-1}\\
			&=\left(
			\begin{pmatrix}
				v_{i-1}\\
				1
			\end{pmatrix}p_i
			+\begin{pmatrix}
				v_{i}\\
				1
			\end{pmatrix}q_i
			\right)
			(1-q_i)^{-1}\eta_i^{-1}.
		\end{align*}
		From the second coordinate of the $2$-dimensional vector, one directly obtains  $\eta_i^{-1}=(1-q_i)(p_i+q_i)^{-1}$,
		and thus $v_i^+$ satisfies \eqref{eq_v+}.
		
	\end{proof}
	
	According to the definition of the leapfrog map, we know that 
	\begin{align}\label{evo_vi-}
		Tv_i^-=v_i,\quad Tv_i=v_i^+.
	\end{align}
	Moreover, from the compatibility conditions of \eqref{evo_vi-} and \eqref{eq_v+},
	we have the following proposition.
	\begin{theorem}
		In the coordinates $\{p_i,q_i\}_{i\in\mathbb{Z}}$,
		the non-commutative leapfrog map is given by the formulas
		\begin{subequations}
			\begin{align}
				T(q_i)&=q_i^+=(p_i+q_i)q_{i+1}(p_{i+1}+q_{i+1})^{-1},\label{sub1}\\
				T(p_i)&=p_i^+=h_i^{-1}p_ih_{i+1},\label{sub2}
			\end{align}
		\end{subequations}
		where $h_i:=(p_{i-1}+q_{i-1})^{-1}-q_i(p_i+q_i)^{-1}$.
	\end{theorem}
	\begin{proof}
		Substituting \eqref{eq_v-} into the equation \eqref{evo_vi-},
		we have
		\begin{align*}
			T(v_i^-)=T\left((v_{i+1}-v_iq_i)(1-q_i)^{-1} \right)
			=(v_{i+1}^+-v_i^+q_i^+)(1-q_i^+)^{-1}=v_i.
		\end{align*}
		Eliminating $v_i^+$, $v_{i+1}^+$ by using \eqref{eq_v+},
		we get
		\begin{align*}
			v_{i+1}q_{i+1}(p_{i+1}+q_{i+1})^{-1}+v_i\left(p_{i+1}(p_{i+1}+q_{i+1})^{-1}-q_i(p_i+q_i)^{-1}q_i^+\right)\\
			=v_i(1-q_i^+)^{-1}+v_{i-1}p_i(p_i+q_i)^{-1}q_i^+.
		\end{align*}
		Moreover, by substituting \eqref{recur_v} into the above equation,
		it yields
		\begin{align*}
			v_i\left((1-p_i)q_{i+1}(p_{i+1}+q_{i+1})^{-1}+p_{i+1}(p_{i+1}+q_{i+1})^{-1}-q_i(p_i+q_i)^{-1}q_i^+  \right)\\
			=v_i(1-q_i^+)^{-1}-v_{i-1}\left(p_iq_{i+1}(p_{i+1}+q_{i+1})^{-1}-p_i(p_i+q_i)^{-1}q_i^+  \right).
		\end{align*}
		Since $v_i$ and $v_{i-1}$ are two independent points, 
		we get \eqref{sub1}.
		
		Next, the action of the map $T$ on \eqref{recur_v} results in
		\begin{align*}
			v_{i+1}^+=v_i^+(1-p_i^+)+v_{i-1}^+p_i^+.
		\end{align*}
		Substituting \eqref{eq_v+} into the equation,
		we get
		\begin{align*}
			&\left(v_ip_{i+1}+v_{i+1}q_{i+1}\right)(p_{i+1}+q_{i+1})^{-1}\\
			&\quad=\left(v_{i-1}p_i+v_{i}q_i\right)(p_i+q_i)^{-1}
			(1-p_i^+)
			+\left(v_{i-2}p_{i-1}+v_{i-1}q_{i-1}\right)(p_{i-1}+q_{i-1})^{-1}
			p_i^+.
		\end{align*}
		The elimination of $v_{i+1}$ and $v_{i-2}$ by using \eqref{recur_v},
		and the comparison in the coefficients of $v_i$ and $v_{i-1}$
		lead to \eqref{sub2}.
		
	\end{proof}
	The next proposition indicates that the non-commutative leapfrog map is integrable in the sense of admitting a Lax pair.
	\begin{proposition}
		The system \eqref{sub1}-\eqref{sub2} admit the following Lax pair
		\begin{align}
			\begin{aligned}
				&v_{i+1}+v_i\left(p_i+q_i-1\right)
				=z\left(v_{i-1}p_i+v_iq_i\right),
				\\
				&T(v_i)
				=
				\left(v_{i-1}p_i+v_iq_i\right)\left(p_i+q_i\right)^{-1},
			\end{aligned}
		\end{align}
		where $z$ is a spectral parameter commuting with all elements in $\mr$.
	\end{proposition}

	\subsection{Cross-ratio type coordinates}
	In this part, the non-commutative leapfrog map in terms of non-commutative cross-ratio type coordinates is formulated.
	According to Theorem \ref{thm_equi},
	the non-commutative cross-ratios are not  invariants of projective transformations,
	but relative invariants, that is, for all $x,y,z,t\in\mathbb{P}^1$,
	\begin{align*}
		\kappa(gx\lambda_1,gy\lambda_2,gz\lambda_3,gt\lambda_4)
		=\lambda_3^{-1}\kappa(x,y,z,t)\lambda_3.
	\end{align*}
	Therefore,
	it is necessary to make scaling transformations for points in $(S^-,S)$ such that cross-ratios can be suitably scaled. 
	Let's denote a new coordinate of $S_i$ by $u_i=\left(\begin{array}{c}
		u_{i,1}\\
		u_{i,2}\end{array}
	\right)$ and that of $S_i^-$ by $u_i^-=\left(\begin{array}{c}
		u_{i,1}^-\\
		u_{i,2}^{-}\end{array}
	\right)$.
	If we take $u_i^-$ and $u_{i-1}^-$ as a basis in $\mathcal{R}^2$,
	then coordinates of $u_i$ and $u_{i+1}$ could be alternatively written as  right linear combinations of $u_i^-$ and $u_{i-1}^-$.
	For simplicity,
	we assume that	
	\begin{align}\label{equ1}
		u_i=u_{i+1}^-+u_i^-a_i,\quad u_{i+1}=u_{i+1}^--u_i^-b_i,
	\end{align}
	where $a_i,b_i\in\mr\backslash\{0\}$ and $a_i+b_i\ne0$.
	This condition could be expressed in terms of quasi-determinants. 
	\begin{proposition}
		For any points $u_i,u_{i+1}\in S$ and $u_i^-,u_{i+1}^-\in S^-$, they satisfy the following constraints
		\begin{align}\label{con_det}
			\left|\begin{array}{cc}
				u_{i,1}^-&u_{i+1,1}^-\\
				u_{i,2}^-&\boxed{u_{i+1,2}^-}
			\end{array}
			\right|=\left|\begin{array}{cc}
				u_{i,1}^-&u_{i,1}\\
				u_{i,2}^-&\boxed{u_{i,2}}
			\end{array}
			\right|=\left|\begin{array}{cc}
				u_{i,1}^-&u_{i+1,1}\\
				u_{i,2}^-&\boxed{u_{i+1,2}}
			\end{array}
			\right|.
		\end{align}
	\end{proposition}
	\begin{proof}
		As $\left(\begin{array}{c}
			u_{i,1}\\
			u_{i,2}
		\end{array}
		\right)$ is a right linear combination of  $\left(\begin{array}{c}
			u_{i,1}^{-}\\
			u_{i,2}^{-}
		\end{array}
		\right)$ and  $\left(\begin{array}{c}
			u_{i+1,1}^-\\
			u_{i+1,2}^-
		\end{array}
		\right)$, it is known from Proposition \ref{p-equi} that 
		\begin{align*}
			\left|\begin{array}{cc}
				u_{i,1}^-&u_{i,1}\\
				u_{i,2}^-&\boxed{u_{i,2}}
			\end{array}
			\right|=\left|\begin{array}{cc}
				u_{i,1}^-&u_{i+1,1}^-+u_{i,1}^-a_i\\
				u_{i,2}^-&\boxed{u_{i+1,2}^-+u_{i,2}^-a_i}
			\end{array}
			\right|=\left|\begin{array}{cc}
				u_{i,1}^-&u_{i+1,1}^-\\
				u_{i,2}^-&\boxed{u_{i+1,2}^-}
			\end{array}
			\right|,
		\end{align*}
		and the rest can be similarly verified.
	\end{proof}
	\begin{remark}
		It is known from \eqref{equ1} that  $u_{i+1}$ could be written as a right linear combination of $u_i$ and $u_i^-$, which means that 
		\begin{align}\label{equ3}
			u_{i+1}=u_i+u_i^-c_i, \quad c_i\in\mr^\times.
		\end{align}
		Moreover, $c_i$ satisfies $a_i+b_i+c_i=0$.
	\end{remark}
The existence for a scaling transformation for coordinates is demonstrated by the following proposition.
	\begin{proposition}
		There exists a scaling transformation
		\begin{align}\label{uv_scal}
			u_i^-=
			\begin{pmatrix}
				v_i^-\\
				1
			\end{pmatrix}
			(V_i^-)^{-1},\quad
			u_i=
			\begin{pmatrix}
				v_i\\
				1
			\end{pmatrix}
			V_i^{-1},
		\end{align}
		where $V_i, V_i^-\in\mr^\times$,  
		such that coordinates $u_i$ and $u_i^-$ satisfy \eqref{con_det}.
		
	\end{proposition}
	\begin{proof}
		According to \eqref{eq_v+},
		we have
		\begin{align*}
			v_{i+1}
			=\left(v_{i}^-p_{i+1}^-+v_{i+1}^-q_{i+1}^-\right)(p_{i+1}^-+q_{i+1}^-)^{-1}.
		\end{align*}
		By substituting \eqref{uv_scal} into it,
		we get
		\begin{align}\label{tr1}
			u_{i+1}V_{i+1}=\left(u_{i}^-V_i^-p_{i+1}^-+u_{i+1}^-V_{i+1}^-q_{i+1}^-\right)(p_{i+1}^-+q_{i+1}^-)^{-1}.
		\end{align}
		Moreover, according to \eqref{eq_v-}, there holds that
		\begin{align}\label{tr2}
			u_i^-V_i^-
			=\left(u_{i+1}V_{i+1}
			-u_iV_iq_i\right)
			(1-q_i)^{-1}.
		\end{align}
		Comparing \eqref{tr1}, \eqref{tr2} with \eqref{equ1} and \eqref{equ3},
		we have
		\begin{align*}
			V_{i+1}=V_{i+1}^-q_{i+1}^-(p_{i+1}^-+q_{i+1}^-)^{-1},\quad
			V_{i+1}=V_iq_i.
		\end{align*}
		This is a linear system of $V_i$,
		and its compatibility condition is just \eqref{sub1}.
		Thus there exists a solution for $V_i$.
		
	\end{proof}
	
	\begin{theorem}\label{thm1.8}
		The coordinates $a_i$, $b_i$ in \eqref{equ1} could be written in terms of the non-commutative cross-ratios
		\begin{subequations}
			\begin{align}
				&a_i
				=\kappa(u_{i-1}^-,u_{i+1}^-,u_i^-,u_i ),\label{cross_a}\\
				&b_i
				=-\kappa(u_i,u_{i+1},u_i^-,u_{i+1}^-)
				\kappa(u_{i-1}^-,u_{i+1}^-,u_i^-,u_i ).\label{cross_b}
			\end{align}
		\end{subequations}
		Under these coordinates,
		the non-commutative leapfrog map $T$ defined by Definition \ref{def_leap} can be characterized by
		\begin{subequations}
			\begin{align}
				a_i^+
				&=\left(a_{i-1}+b_{i-1}\right)^{-1}a_{i-1}\left(a_{i}+b_{i}\right),\label{leap_a}\\
				b_i^+
				&=\left(a_{i}+b_{i}\right)^{-1}b_{i}\left(a_{i+1}+b_{i+1}\right).\label{leap_b}
			\end{align}
		\end{subequations}
	\end{theorem}
	\begin{proof}
		We give a proof by using the properties of non-commutative cross-ratios.
		We know from \eqref{nccr} that the non-commutative cross-ratio could be written in terms of quasi-Pl\"ucker coordinates and
		\begin{align*}
			\kappa(u_{i-1}^-,u_{i+1}^-,u_i^-,u_i)=q_{u_i^-,u_i}^{u_{i+1}^-}q_{u_i,u_i^-}^{u_{i-1}^-}.
		\end{align*}
		Moreover, by the properties of quasi-Pl\"ucker coordinate \eqref{alt} and \eqref{quasi1}, we obtain
		\begin{align*}
			q_{u_i^-,u_i}^{u_{i+1}^-}=\left|\begin{array}{cc}
				u_{i+1,1}^-&u_{i,1}^-\\
				u_{i+1,2}^-&\boxed{u_{i,2}^-}
			\end{array}
			\right|^{-1}\left|\begin{array}{cc}
				u_{i+1,1}^-&u_{i,1}\\
				u_{i+1,2}^-&\boxed{u_{i,2}}
			\end{array}
			\right|=a_i,
		\end{align*}
		where \eqref{con_det} is used.
		Similarly, by realizing that $u_i$ could be written as a right linear combination of $u_{i-1}^-$ and $u_i^-$, 
		we can get 
		$
		q_{u_i,u_i^-}^{u_{i-1}^-}=1
		$
		and \eqref{cross_a} is valid. Moreover, one can show that $q_{u_i^-,u_{i+1}^-}^{u_{i+1}}=b_i$ and $q_{u_{i+1}^-,u_i^-}^{u_i}=a_i^{-1}$, from which \eqref{cross_b} is directly verified.
		
		To characterize the leapfrog map $T$, we need to replace points in $(S^-,S)$ by $(S,S^+)$. In this case,	$a_i^+$ and $b_i^+$ have expressions
		\begin{align*}
			a_i^+=\kappa(u_{i-1},u_{i+1},u_i,u_i^+), \quad b_i^+=-\kappa(u_i^+,u_{i+1}^+,u_i,u_{i+1})\kappa(u_{i-1},u_{i+1},u_i,u_i^+).
		\end{align*}
		Moreover, since $T$ is a leapfrog map, then
		there exists a projective transformation
		$\tilde{g}_i\in\text{GL}_2(\mr)$ such that
		$$\tilde{g}_i:(u_{i-1},u_{i+1},u_i,u_i^+)\mapsto(u_{i+1}\tilde{\alpha}_i,u_{i-1}\tilde{\beta}_i,-u_i,u_i^-\tilde{\eta}_i),\quad 
		\tilde{\alpha}_i,\,\tilde{\beta}_i,\,\tilde{\eta}_i\in\mr^\times.$$
		By \eqref{cross_pro}, we know that $a_i^+$ has an alternative expression 
		\begin{align*}
			a_i^+=\kappa(u_{i-1},u_{i+1},u_i,u_i^+)=\kappa(u_{i+1},u_{i-1},u_i,u_i^-),
		\end{align*}
		which is helpful in the verifications of \eqref{leap_a} and \eqref{leap_b}. By using properties of non-commutative cross-ratios \eqref{eq_1-} and \eqref{eq_3}, we have
		\begin{align*}
			a_i+b_i&=(1-\kappa(u_i,u_{i+1},u_i^-,u_{i+1}^-))
			\kappa(u_{i-1}^-,u_{i+1}^-,u_i^-,u_i )\\
			&=\kappa(u_{i+1}^-,u_{i+1},u_i^-,u_i)\kappa(u_{i-1}^-,u_{i+1}^-,u_i^-,u_i )\\
			&=\kappa(u_{i-1}^-,u_{i+1},u_i^-,u_i).
		\end{align*}
		Moreover, 
		\begin{align*}
			\begin{split}
				&\left(a_{i-1}+b_{i-1}\right)^{-1}a_{i-1}\left(a_{i}+b_{i}\right)\\
				=&\kappa(u_i,u_{i-2}^-,u_{i-1}^-,u_{i-1})\kappa(u_{i-2}^-,u_i^-,u_{i-1}^-,u_{i-1})
				\kappa(u_{i-1}^-,u_{i+1},u_i^-,u_i)\\
				=&q_{u_{i-1}^-,u_{i-1}}^{u_{i-2}^-}\kappa(u_{i-2}^-,u_i,u_{i-1},u_{i-1}^-)q_{u_{i-1},u_{i-1}^-}^{u_{i-2}^-}\\
				&\cdot q_{u_{i-1}^-,u_{i-1}}^{u_{i-2}^-}
				\kappa(u_{i}^-,u_{i-2}^-,u_{i-1},u_{i-1}^-)
				q_{u_{i-1},u_{i-1}^-}^{u_{i-2}^-}
				\kappa(u_{i-1}^-,u_{i+1},u_i^-,u_i)\\
				=&\kappa(u_{i}^-,u_i,u_{i-1},u_{i-1}^-)\kappa(u_{i-1}^-,u_{i+1},u_i^-,u_i),
			\end{split}
		\end{align*}
		where in the last step $q_{u_{i-1}^-,u_{i-1}}^{u_{i-2}^-}=q_{u_{i-1},u_{i-1}^-}^{u_{i-2}^-}= 1$ and equations \eqref{eq_3}, \eqref{eq_per} are used. We can further show that
		\begin{align*}
			\kappa(u_{i}^-,u_i,u_{i-1},u_{i-1}^-)\kappa(u_{i-1}^-,u_{i+1},u_i^-,u_i)=\kappa(u_{i+1},u_{i-1},u_i,u_i^-)
			=a_i^+,
		\end{align*}
		and \eqref{leap_a} is then verified. \eqref{leap_b} could be similarly verified.

	\end{proof}
	
	\subsection{Y-systems}
	It is known that many pentagram-type maps could be described as $Y$-mutations in cluster algebra,
	thus establishing connections between cluster algebra and projective geometry.
	Here we construct a non-commutative $Y$-system related to the non-commutative leapfrog map given in Definition \ref{def_leap}.
	In the following,
	we use the notation
	\begin{align*}
		u_i^-\rightarrow u_i^j,\quad
		u_i\rightarrow u_i^{j+1},\quad
		u_i^+\rightarrow
		u_i^{j+2}
	\end{align*}
	to regard each leapfrog map as a step of discrete-time evolution.
	Let's define $y$-variable as a non-commutative cross-ratio
	\begin{align}\label{def_y}
		y_i^j=y_i^j(u)=-\kappa(u_i^{j+1},\,u_{i+1}^{j+1},\,u_i^ja_i^j,\,u_{i+1}^j).
	\end{align}
	
	\begin{theorem}
		The $y$-variable
		satisfies
		\begin{align}\label{eq_y}
			\begin{split}
				(b_i^j)^{-1}
				\left(1+y_{i-1}^{j}\right)
				\left(1+\left(y_i^j\right)^{-1}\right)^{-1}
				\left(y_i^{j-1}\right)^{-1}b_i^j\\
				=
				\left(1+\left(y_i^j\right)^{-1}\right)
				y_i^{j+1}
				\left(1+y_{i+1}^j\right)^{-1},
			\end{split}
		\end{align}
		where $b_i^j$ is given by \eqref{cross_b}.
	\end{theorem}
	\begin{proof}
		Using the property of non-commutative cross-ratios and Theorem \ref{thm1.8},
		we have
		\begin{align*}
			y_i^j=-(a_i^j)^{-1}
			\kappa(u_i^{j+1},\,u_{i+1}^{j+1},\,u_i^j,\,u_{i+1}^j)
			a_i^j=
			(a_i^j)^{-1}b_i^j,
		\end{align*}
		or equivalently,
		\begin{align*}
			b_i^j=a_i^jy_i^j.
		\end{align*}
		Substituting it into \eqref{leap_a} and \eqref{leap_b} yields
		\begin{subequations}
			\begin{align}
				&a_{i+1}^j
				=(y_i^{j-1})^{-1}a_i^jy_i^j,\label{aij1}\\
				&a_i^{j+1}
				=\left(1+y_{i-1}^j\right)^{-1}a_i^j\left(1+y_i^j\right).\label{aij2}
			\end{align}
		\end{subequations}
		There are two ways to compute $a_{i+1}^{j+1}$.
		From \eqref{aij1}, we have
		\begin{align*}
			a_{i+1}^{j+1}
			=(y_i^{j})^{-1}a_i^{j+1}y_i^{j+1}
			=(y_i^{j})^{-1}\left(1+y_{i-1}^j\right)^{-1}a_i^j\left(1+y_i^j\right)y_i^{j+1},
		\end{align*}
		and \eqref{aij2} leads to
		\begin{align*}
			a_{i+1}^{j+1}
			=\left(1+y_{i}^j\right)^{-1}a_{i+1}^j\left(1+y_{i+1}^j\right)
			=\left(1+y_{i}^j\right)^{-1}(y_i^{j-1})^{-1}a_i^jy_i^j\left(1+y_{i+1}^j\right).
		\end{align*}
		Therefore, from the compatibility condition, we get
		\begin{align*}
			(y_i^{j})^{-1}\left(1+y_{i-1}^j\right)^{-1}a_i^j\left(1+y_i^j\right)y_i^{j+1}
			=
			\left(1+y_{i}^j\right)^{-1}(y_i^{j-1})^{-1}a_i^jy_i^j\left(1+y_{i+1}^j\right).
		\end{align*}
		The substitution of the identity $a_i^j=b_i^j(y_i^j)^{-1}$ yields \eqref{eq_y}.

	\end{proof}
	
	In the commutative case,
	the equation \eqref{eq_y} reduces to
	\begin{align}\label{eq_com_y}
		y_i^{j+1}y_i^{j-1}
		=\frac{\left(1+y_{i+1}^j\right)\left(1+y_{i-1}^j\right)}{\left(1+\left(y_i^j\right)^{-1}\right)^2},
	\end{align}
	which is the $Y$-system corresponding to the commutative leapfrog map \cite{Y2016}.
	
	\section{Non-commutative Laurent bi-orthogonal polynomials and non-commutative relativistic Toda equation}\label{sec_ops}
	
	In this section, Lax representation and continuum limit of equations \eqref{leap_a} and \eqref{leap_b} are obtained by utilizing the non-commutative Laurent bi-orthogonal polynomials. Non-commutative orthogonal polynomials theory was proposed in \cite[Sec. 8.3]{gelfand95} as a formal analog of the matrix-valued orthogonal polynomials, and its connection with non-commutative Toda lattice through spectral transformations was recently considered in \cite{li23}. Let $\mr$ be a skew field generated by unity $1$ and formal moments $\{m_i\}_{i=-\infty}^\infty$, and denote by $\mr[[z]]$ (resp. $\mr[z]$) the ring of formal power  series (resp. polynomial) in $z$ with coefficients in the skew field $\mr$.
	This skew field is endowed with an involution $\mr\to \mr^*$ such that $(a_i)^*=a_i^*$, which could be extended to the polynomial ring by $\mr[z]\to \mr^*[z^{-1}]$ such that
	\begin{align*}
		\left(
		\sum_{i}a_iz^i
		\right)^*=\sum_{i}a_i^* z^{-i}.
	\end{align*}
	Then we can define an inner product $\langle\cdot,\cdot\rangle: \,\mr[[z]]\times \mr[[z]]\to \mr$ and
	\begin{align}\label{inn}
		\langle \sum_i a_iz^i,\sum_j b_jz^j\rangle=\sum_{i,j}a_i m_{i-j}b_j^*.
	\end{align}	
	Moreover, we have the following properties for the inner product.
	\begin{proposition}
		Given $f_1(z),f_2(z),g_1(z),g_2(z)\in \mr[[z]]$ and $\alpha_1,\alpha_2,\beta_1,\beta_2\in \mr$, we have
		\begin{enumerate}
			\item $\langle \alpha_1f_1(z)+\alpha_2f_2(z),g(z)\rangle=\alpha_1\langle f_1(z),g(z)\rangle+\alpha_2\langle f_2(z),g(z)\rangle$;
			\item $\langle f(z),\beta_1g_1(z)+\beta_2g_2(z)\rangle=\langle f(z),g_1(z)\rangle \beta_1^*+\langle f(z),g_2(z)\rangle\beta_2^*$;
			\item $\langle zf(z),g(z)\rangle=\langle f(z),z^{-1}g(z)\rangle$.
		\end{enumerate}
		
	\end{proposition}	
	Therefore, there are monic non-commutative Laurent bi-orthogonal polynomials $\{P_n(z)\}_{n=0}^\infty$ and $\{Q_n(z)\}_{n=0}^\infty$ with respect to the inner product \eqref{inn} by 
	\begin{align}\label{or}
		\langle P_n(z),Q_m(z)\rangle=H_n\delta_{n,m},
	\end{align}
	where $H_n\in\mr^\times$ is a normalization factor. We call $P_n(z)$ as a monic polynomial if its coefficient in the highest order is the unity in $\mr$, i.e. $P_n(z)$ admits the form $1\cdot z^n+\xi_{n,n-1}\cdot z^{n-1}+\cdots+\xi_{n,0}\cdot z^0$ for some $\xi_{n,i}\in \mr$.

	\subsection{Quasi-determinant formula for non-commutative Laurent bi-orthogonal polynomials}
	In this part, we first demonstrate the quasi-determinant formula for non-commutative Laurent bi-orthogonal polynomials defined by \eqref{or}.
	\begin{theorem}\label{qor}
		If the Toeplitz matrix $(m_{i-j})_{i,j=0}^n$ is invertible for all $n\in\mathbb{N}$, then non-commutative Laurent bi-orthogonal polynomials $\{P_n(z)\}_{n=0}^\infty$ and $\{Q_n(z)\}_{n=0}^\infty$ defined by \eqref{or} have the following quasi-determinant expressions
		\begin{align}\label{form}
			P_n(z)=\left|\begin{array}{cccc}
				m_0&\cdots&m_{1-n}&1\\
				\vdots&&\vdots&\vdots\\
				m_{n-1}&\cdots&m_0&z^{n-1}\\
				m_n&\cdots&m_1&\boxed{z^n}
			\end{array}
			\right|, \quad (Q_n(z))^*=\left|\begin{array}{cccc}
				m_0&\cdots&m_{1-n}&m_{-n}\\
				\vdots&&\vdots&\vdots\\
				m_{n-1}&\cdots&m_0&m_{-1}\\
				1&\cdots&z^{1-n}&\boxed{z^{-n}}
			\end{array}
			\right|,
		\end{align}
		and the normalization factor admits the following Toeplitz quasi-determinant formula
		\begin{align*}
			H_n=\left|\begin{array}{cccc}
				m_0&\cdots&m_{1-n}&m_{-n}\\
				\vdots&&\vdots&\vdots\\
				m_{n-1}&\cdots&m_0&m_{-1}\\
				m_n&\cdots&m_1&\boxed{m_0}
			\end{array}
			\right|.
		\end{align*}
	\end{theorem}
	\begin{proof}
		Firstly, the orthogonality \eqref{or} could be equivalently written by 
		\begin{align*}
			\langle P_n(z),z^i\rangle=0,\quad 0\leq i\leq n-1,
		\end{align*}
		which is a linear system with non-commutative variables
		\begin{align*}
			m_j+\xi_{n,n-1}m_{j-1}+\cdots+\xi_{n,0}m_{j-n}=0,\quad j=1,\cdots,n.
		\end{align*}
		According to Proposition \ref{p-ls}, we know that if the coefficient matrix $(m_{j-i})_{j,i=1}^n$ is invertible, then corresponding linear system has the unique solution
		\begin{align*}
			\xi_{n,i}=\left|\begin{array}{cccc}
				m_0&\cdots&m_{1-n}&\\
				\vdots&&\vdots&e_i^\top\\
				m_{n-1}&\cdots&m_0&\\
				m_n&\cdots&m_1&\boxed{0}
			\end{array}
			\right|,\quad i=0,1,\cdots,n-1,
		\end{align*}
		where $e_i$ is a unit row vector whose $(i+1)$-th position is the unity. Therefore, according to the linearity, we know that $P_n(z)$ satisfies \eqref{form}, and $(Q_n(z))^*$ could be similarly verified.
		Moreover, 
		\begin{align*}
			H_n=\langle P_n(z),Q_n(z)\rangle=\langle P_n(z),z^{n}\rangle=\left|\begin{array}{cccc}
				m_0&\cdots&m_{1-n}&m_{-n}\\
				\vdots&&\vdots&\vdots\\
				m_{n-1}&\cdots&m_0&m_{-1}\\
				m_n&\cdots&m_1&\boxed{m_0}
			\end{array}
			\right|,
		\end{align*}
		which could be regarded as a non-commutative version of the Toeplitz determinant.
	\end{proof}
	
	\subsection{Spectral transformations for non-commutative Laurent bi-orthogonal polynomials}
	To demonstrate spectral transformations, we need to introduce notations of adjacent families for non-commutative Laurent bi-orthogonal polynomials. 	
	
	\begin{definition}
		For $k\in\mathbb{Z}$, let's define a discrete deformed inner product $\langle \cdot,\cdot\rangle_k$: $\mr[[z]]\times \mr[[z]]\to \mr$ such that
		\begin{align*}
			\langle \sum_i a_iz^i,\sum_j b_jz^j\rangle_k=a_i m_{i-j+k}b_j^*.
		\end{align*} 
		Moreover, if 
		\begin{align}\label{or2}
			\langle P_n^{(k)}(z),Q_m^{(k)}(z)\rangle_k=H_n^{(k)}\delta_{n,m}
		\end{align}
		is valid for some invertible normalization factor $H_n^{(k)}$, 
		then we call $\{P_n^{(k)}(z)\}_{k=0}^\infty$ and $\{Q_n^{(k)}(z)\}_{n=0}^\infty$ $(k\in\mathbb{Z})$ as adjacent families of the non-commutative Laurent bi-orthogonal polynomials.
	\end{definition}
	Similar to the Theorem \ref{qor}, one could show from the orthogonal relation \eqref{or2} that
	\begin{align*}
		P_n^{(k)}(z)=\left|\begin{array}{cccc}
			m_k&\cdots&m_{k+1-n}&1\\
			\vdots&&\vdots&\vdots\\
			m_{k+n-1}&\cdots&m_{k}&z^{n-1}\\
			m_{k+n}&\cdots&m_{k+1}&\boxed{z^n}
		\end{array}
		\right|,\quad (Q_n^{(k)}(z))^*=\left|\begin{array}{cccc}
			m_{k}&\cdots&m_{k-n+1}&m_{k-n}\\
			\vdots&&\vdots&\vdots\\
			m_{k+n-1}&\cdots&m_k&m_{k-1}\\
			1&\cdots&z^{1-n}&\boxed{z^{-n}}
		\end{array}
		\right|,
	\end{align*}
	and $H_n^{(k)}$ is a shifted Toeplitz quasi-determinant
	\begin{align*}
		H_n^{(k)}=\left|\begin{array}{ccc}
			m_k&\cdots&m_{k-n}\\
			\vdots&&\vdots\\
			m_{n-k}&\cdots&\boxed{m_k}
		\end{array}
		\right|.
	\end{align*}
	Firstly, with adjacent families of non-commutative Laurent bi-orthogonal polynomials, we have the following Christoffel transformation for $(Q_n^{(k)}(z))^*$.
	\begin{proposition}
		There exists the following spectral transformation between $\{(Q_n^{(k)}(z))^*\}_{n\in\mathbb{N}}$ and $\{(Q_n^{(k+1)}(z))^*\}_{n\in\mathbb{N}}$
		\begin{align}\label{ct}
			(Q_{n+1}^{(k)}(z))^*=z^{-1}(Q_n^{(k)}(z))^*-(Q_n^{(k+1)}(z))^*\varphi_n^{(k)},
		\end{align}
		where 
		\begin{align}\label{varphi}
			\varphi_n^{(k)}=\left|\begin{array}{ccc}
				m_k&\cdots&\boxed{m_{k-n}}\\
				\vdots&&\vdots\\
				m_{k+n}&\cdots&m_k
			\end{array}
			\right|^{-1}\left|\begin{array}{ccc}
				m_{k-1}&\cdots&\boxed{m_{k-1-n}}\\
				\vdots&&\vdots\\
				m_{k-1+n}&\cdots&m_{k-1}
			\end{array}
			\right|.
		\end{align}
	\end{proposition}
	\begin{proof}
		This proof is based on the non-commutative Jacobi identity and homological relations for quasi-determinants. By acting non-commutative Jacobi identity \eqref{ncj1} to the $(1,n+2)$-rows and $(1,n+2)$-columns to $(Q_{n+1}^{k}(z))^*$, we have
		\begin{align}\label{ct1}
			\begin{aligned}
				&\left|\begin{array}{cccc}
					m_k&\cdots&m_{k-n}&m_{k-n-1}\\
					\vdots&&\vdots&\vdots\\
					m_{k+n}&\cdots&m_k&m_{k-1}\\
					1&\cdots&z^{-n}&\boxed{z^{-n-1}}
				\end{array}
				\right|=\left|\begin{array}{ccc}
					m_k&\cdots&m_{k-n}\\
					\vdots&&\vdots\\
					m_{k+n-1}&\cdots&m_{k-1}\\
					z^{-1}&\cdots&\boxed{z^{-n-1}}
				\end{array}
				\right|\\
				&-\left|\begin{array}{ccc}
					m_{k+1}&\cdots&m_{k+1-n}\\
					\vdots&&\vdots\\
					m_{k+n-1}&\cdots&m_k\\
					\boxed{1}&\cdots&z^{-n}
				\end{array}
				\right|\left|\begin{array}{ccc}
					\boxed{m_k}&\cdots&m_{k-n}\\
					\vdots&&\vdots\\
					m_{k+n}&\cdots&m_k
				\end{array}
				\right|^{-1}\left|\begin{array}{ccc}
					m_{k-1}&\cdots&\boxed{m_{k-1-n}}\\
					\vdots&&\vdots\\
					m_{k-1+n}&\cdots&m_{k-1}
				\end{array}
				\right|.
			\end{aligned}
		\end{align}
		Moreover, by using homological relations \eqref{hm1} and \eqref{hm2}, we have
		\begin{align*}
			&\left|\begin{array}{ccc}
				m_{k+1}&\cdots&m_{k+1-n}\\
				\vdots&&\vdots\\
				m_{k+n-1}&\cdots&m_k\\
				\boxed{1}&\cdots&z^{-n}
			\end{array}
			\right|=\left|\begin{array}{ccc}
				m_{k+1}&\cdots&m_{k+1-n}\\
				\vdots&&\vdots\\
				m_{k+n}&\cdots&m_{k}\\
				1&\cdots&\boxed{z^{-n}}
			\end{array}
			\right|\left|\begin{array}{ccc}
				m_{k+1}&\cdots&m_{k+1-n}\\
				\vdots&&\vdots\\
				m_{k+n}&\cdots&m_{k}\\
				\boxed{0}&\cdots&1
			\end{array}
			\right|,\\
			&\left|\begin{array}{ccc}
				\boxed{m_k}&\cdots&m_{k-n}\\
				\vdots&&\vdots\\
				m_{k+n}&\cdots&m_k
			\end{array}
			\right|^{-1}=\left(
			\left|\begin{array}{ccc}
				m_k&\cdots&\boxed{m_{k-n}}\\
				\vdots&&\vdots\\
				m_{k+n}&\cdots&m_k
			\end{array}
			\right|
			\left|\begin{array}{ccc}
				\boxed{0}&\cdots&1\\
				m_{k+1}&\cdots&m_{k+1-n}\\
				\vdots&&\vdots\\
				m_{k+n}&\cdots&m_k
			\end{array}
			\right|
			\right)^{-1}.
		\end{align*}
		Taking the substitution of this equation into \eqref{ct1}, we know that \eqref{ct} is valid for all adjacent families of non-commutative Laurent bi-orthogonal polynomials. 
	\end{proof}
	Besides, there is another spectral transformation for non-commutative Laurent bi-orthogonal polynomials called Geronimus transformation. We have the following proposition.
	\begin{proposition}
		For adjacent families $\{(Q_n^{(k-1)}(z))^*\}_{n=0}^\infty$ and $\{(Q_n^{(k)}(z))^*\}_{n=0}^\infty$, they satisfy the following Geronimus transformation
		\begin{align}\label{gt}
			z^{-1}(Q_n^{(k-1)}(z))^*=(Q_{n+1}^{(k)}(z))^*+(Q_n^{(k)}(z))^*\psi_n^{(k)},
		\end{align}
		where \begin{align}\label{psi}
			\psi_n^{(k)}=\left(
			H_n^{(k)}
			\right)^{-1}H_n^{(k-1)},
		\end{align} and $H_n^{(k)}$ is the shifted Toeplitz quasi-determinant defined in \eqref{or2}.
	\end{proposition}
	\begin{proof}
		The proof of Geronimus transformation is based on the idea of Fourier expansion. Since any  polynomial in $(R[z])^*=R^*[z^{-1}]$ can be written in terms of a right linear combination of $\{(Q_n^{(k)}(z))^*\}_{n=0}^\infty$, we have
		\begin{align}\label{fourier}
			z^{-1}(Q_n^{(k-1)}(z))^*=(Q_{n+1}^{(k)}(z))^*+\sum_{i=0}^n (Q_i^{(k)}(z))^*\psi_{n,i}^{(k)}.
		\end{align}
		Based on the orthogonality of adjacent families \eqref{or2}, it is known that
		\begin{align*}
			\langle z^i,Q_n^{(k-1)}(z)\rangle_{k-1}=\langle z^{i-1},Q_n^{(k-1)}(z)\rangle_k=\langle z^i,zQ_n^{(k-1)}(z)\rangle_k	=0,\quad 0\leq i\leq n-1.
		\end{align*}
		Therefore, by taking \eqref{fourier} into the above formula, we know that there are only two terms left and 
		\begin{align*}
			z^{-1}(Q_n^{(k-1)}(z))^*=(Q_{n+1}^{(k)}(z))^*+(Q_n^{(k)}(z))^*\psi_{n}^{(k)},
		\end{align*}
		where $$\psi_{n}^{(k)}=\langle P_n^{(k)}(z),Q_n^{(k)}(z)\rangle_k^{-1}\langle P_n^{(k)}(z),Q_n^{(k-1)}(z)\rangle_{k-1}=\left(H_n^{(k)}\right)^{-1}H_n^{(k-1)}.$$
	\end{proof}
	\begin{remark}
		It should be noted that this Geronimus transformation could also be directly obtained via a non-commutative Jacobi identity by acting on $(n+1,n+2)$-rows and $(1,n+2)$-columns of $(Q_{n+1}^{(k)}(z))^*$.
	\end{remark}
	
	From the Christoffel transformation \eqref{ct} and Geronimus transformations \eqref{gt},  the following recurrence relation for $(Q_n^{(k)}(z))^*$ is found.
	\begin{theorem}\label{rrrr}
		For each $k\in\mathbb{Z}$ and $n\in\mathbb{N}$, there exists a three-term recurrence relation
		\begin{align}\label{rrr}
			z\left(
			(Q_{n+1}^{(k)}(z))^*
			+(Q_n^{(k)}(z))^*\psi_n^{(k)}
			\right)=(Q_n^{(k)}(z))^*+(Q_{n-1}^{(k)}(z))^*\left(
			\psi_{n-1}^{(k)}-\varphi_{n-1}^{(k-1)}
			\right).
		\end{align}
	\end{theorem}
	If we denote $\xi_i^{(k)}=\psi_i^{(k)}-\varphi_i^{(k-1)}$, and $R^{(k)}=((Q_0^{(k)}(z))^*,(Q_1^{(k)}(z))^*,\cdots)$, then \eqref{rrr} could be written in terms of 
	\begin{align}\label{rrmatrix}
		z^{-1}R^{(k)}=R^{(k)}(I+\Xi^{(k)}\Lambda)(\Psi^{(k)}+\Lambda^{-1}):=R^{(k)}X^{(k)}
	\end{align}
	where $\Xi^{(k)}=\text{diag}(\xi_0^{(k)},\xi_1^{(k)},\cdots)$, $\Psi^{(k)}=(\psi_0^{(k)},\psi_1^{(k)},\cdots)$ and $\Lambda$ represents the shift operator. On the other hand, the Geronimus transformation \eqref{gt} could be written in a matrix form 
	\begin{align}\label{gtmatrix}
		z^{-1}R^{(k-1)}=R^{(k)}Y^{(k)}, \quad Y^{(k)}=(\Psi^{(k)}+\Lambda^{-1}).
	\end{align}
	Therefore, the compatibility condition of \eqref{rrmatrix} and \eqref{gtmatrix} results in the discrete Lax equation 
	\begin{align*}
		X^{(k)}Y^{(k)}=Y^{(k)}X^{(k-1)},
	\end{align*}	
	which is indeed the non-commutative relativistic Toda equation.
	These results are summarized into the following theorem.
	\begin{theorem}
		The non-commutative relativistic Toda equation admits the form
		\begin{align*}
			&(\psi_i^{(k+1)}-\varphi_i^{(k)})\psi_{i+1}^{(k)}=\psi_i^{(k+1)}(\psi_i^{(k)}-\varphi_i^{(k-1)}),\\
			&\psi_i^{(k)}-\varphi_i^{(k)}=\psi_{i-1}^{(k)}-\varphi_{i-1}^{(k-1)},
		\end{align*}
		where $\varphi_i^{(k)}$ and $\psi_i^{(k)}$ have the quasi-determinant formulas \eqref{varphi} and \eqref{psi}.
	\end{theorem}
	This equation is related to the non-commutative leapfrog map \eqref{leap_a}, \eqref{leap_b} through 
	\begin{align*}
	a_i=\psi_i^{(k)},\quad
	b_i=\varphi_i^{(k-1)}-\psi_i^{(k)},
	\end{align*}
	and 
	\begin{align*}
		a_i^+=\psi_i^{(k-1)},\quad
		b_i^+=\varphi_i^{(k-2)}-\psi_i^{(k-1)}.
	\end{align*}
	
	\subsection{Continuous time evolution and non-commutative semi-discrete relativistic Toda equation}
	In this part, we assume that the skew field $\mr$ is dependent with a continuous time variable $t$ whose derivative commutes with involution $*$ over $\mr$. In other words, if $f(t)\in \mr(t)$, then 
	\begin{align*}
		\frac{d}{dt}((f(t))^*)=\left(
		\frac{d}{dt}f(t)
		\right)^*.
	\end{align*}
	Let's consider a time-deformed inner product $\langle\cdot,\cdot\rangle_t$: $\mr(t)[[z]]\times \mr(t)[[z]]\to \mr(t)$ such that
	\begin{align*}
		\frac{d}{dt}\langle \sum_i a_iz^i,\sum_j b_jz^j\rangle_t=\frac{d}{dt}(\sum_{i,j}a_im_{i-j}b_j^*)=\sum_{i,j}\left(
		\frac{da_i}{dt}m_{i-j}b_j^*+a_i\frac{dm_{i-j}}{dt}b_j^*+a_im_{i-j}\frac{db_j^*}{dt}
		\right).
	\end{align*}
	Then time-dependent non-commutative Laurent bi-orthogonal polynomials can be defined by the bi-orthogonal relation
	\begin{align}\label{tip}
		\langle P_n(z;t),Q_m(z;t)\rangle_t=H_n(t)\delta_{n,m},
	\end{align}
	for invertible $H_n(t)$ for all $t\in\mathbb{R}$.

	Following Theorem \ref{rrrr}, we know that three-term recurrence relation 
	\begin{align}\label{newtr}
		z^{-1}\big(
		(Q_n(z))^*+(Q_{n-1}(z))^*\xi_{n-1}
		\big)=(Q_{n+1}(z))^*+(Q_n(z))^*\psi_n
	\end{align}
	is valid for $(Q_n(z))^*$ with recurrence coefficients $\xi_{n-1}$ and $\psi_n$. In this part, we show that under continuous time evolution, $\xi_{n-1}$ and $\psi_n$ should satisfy non-commutative semi-discrete relativistic Toda equations, as a continuum limit of the non-commutative leapfrog map.
	Usually, there are two different evolutions with respect to semi-discrete relativistic Toda equations, and we discuss them separately.

	\subsubsection{Negative flow}
	
	In negative flow case, we require that 
	\begin{align}\label{nef}
		\frac{d}{dt}m_{i-j}=m_{i-j-1}.
	\end{align}
	In other words, if $f(z),g(z)\in \mr[[z]]$, then 
	\begin{align*}
		\frac{d}{dt}\langle f(z),g(z)\rangle_t=\langle z^{-1}f(z),g(z)\rangle_t=\langle f(z),zg(z)\rangle_t.
	\end{align*}
	\begin{proposition}
		For time-dependent non-commutative Laurent bi-orthogonal polynomial $(Q_n(z;t))^*$, there holds 
		\begin{align}\label{tw2}
			\p_t (Q_n(z;t))^*=\left(
			z^{-1}(Q_{n-1}(z;t))^*-(Q_n(z;t))^*
			\right)\xi_{n-1},
		\end{align}
		where $\xi_{n-1}$ is given in \eqref{newtr}.
	\end{proposition}
	\begin{proof}
		This proof is based on the orthogonal relation 
		\begin{align*}
			\langle z^i, Q_n(z;t)\rangle_t=0,\quad 0\leq i\leq n-1.		\end{align*}
		Taking the derivative with respect to $t$ and a substitution of \eqref{newtr} results in
		\begin{align*}
			&\langle z^i, \p_t Q_n(z;t)\rangle_t+\langle z^i, zQ_n(z;t)\rangle_t\\
			&\quad=\langle z^i,\p_t Q_n(z;t)+Q_{n+1}(z;t)+\psi_n^*Q_n(z;t)-za_{n-1}^* Q_{n-1}(z;t)\rangle_t=0.
		\end{align*}
		From the three-term recurrence relation \eqref{newtr}, this equation could be further reduced to
		\begin{align*}
			\langle z^i, \p_t Q_n(z;t)-za_{n-1}^*Q_{n-1}(z;t)+a_{n-1}^* Q_{n}(z;t)\rangle_t=0,\quad 0\leq i\leq n-1,
		\end{align*}
		where degree of the polynomial 
		\begin{align*}
		\tilde{Q}_n(z;t):=\p_t Q_n(z;t)-za_{n-1}^*Q_{n-1}(z;t)+a_{n-1}^* Q_{n}(z;t)
		\end{align*} is less than $n$. Moreover, if we expand $\tilde{Q}_n(z;t)$ as a right linear combination of $\{Q_j(z;t)\}_{j=0}^{n-1}$, then the coefficients are zeros by the orthogonality. Therefore, $\tilde{Q}_n(z;t)=0$ and \eqref{tw2} is verified.
	\end{proof}
	
	\begin{remark}
		In fact, time evolution \eqref{tw2} could be alternatively written as
		\begin{align}\label{tw}
			\p_t (Q_{n+1}(z;t))^*+\p_t (Q_{n}(z;t))^*\xi_{n}=(Q_{n}(z;t))^*\zeta_{n},
		\end{align}
		where $\zeta_{n}=(\psi_n-\xi_{n-1})\xi_n$, which facilitate the formulation of Lax integrability.
	\end{remark}
	
	By denoting $\Psi=\text{diag}(\psi_0,\psi_1,\cdots)$, $\Xi=\text{diag}(\xi_0,\xi_1,\cdots)$ and $Z=(\zeta_0,\zeta_1,\cdots)$, then \eqref{newtr} and \eqref{tw} could be written in terms of matrix form
	\begin{subequations}
		\begin{align}
			z^{-1}R\mathbb{A}=R\mathbb{B}, \label{w1}\\
			\p_t R \mathbb{A}=R\mathbb{C},\label{w2}
		\end{align}	
	\end{subequations}
	where $R=((Q_0(z;t))^*,(Q_1(z;t))^*,\cdots)$ and 
	\begin{align*}
		\mathbb{A}=(I+\Xi\Lambda),\quad \mathbb{B}=\Psi+\Lambda^{-1},\quad \mathbb{C}=Z\Lambda.
	\end{align*}
	\begin{theorem}
		The compatible condition of \eqref{w1} and \eqref{w2} gives to the non-commutative semi-discrete relativistic Toda equation in negative flow and
		\begin{align}\label{ncrt}
			\p_t (\mathbb{B}\mathbb{A}^{-1})=[\mathbb{B}\mathbb{A}^{-1},\mathbb{C}\mathbb{A}^{-1}].
		\end{align}
	\end{theorem}
	To express this matrix equation into a scalar form, one obtains
	\begin{align*}
		&\p_t \Psi-\Lambda^{-1}\p_t \Xi\Lambda=\Lambda^{-1}Z\Lambda-Z,\\
		&\p_t \Xi\Lambda=-Z\Lambda+(\Lambda^{-1}\Xi\Lambda-\Psi)^{-1}Z\Lambda(\Xi-\Psi),
	\end{align*}	
	which is equal to
	\begin{align*}	
		\p_t \xi_n&=-(\psi_n-\xi_{n-1})\xi_n-\xi_n(\xi_{n+1}-\psi_{n+1}),\\
		\p_t \psi_n&=\xi_{n-1}\psi_n-\psi_n\xi_n.
	\end{align*}
	The semi-discrete non-commutative relativistic Toda equation has been appeared in \cite{casati21,gehktman11} with Hamiltonian structure and recursion operator.

	\subsubsection{Positive flow} 
	With positive flow, we require that
	\begin{align*}
		\frac{d}{dt}m_{i-j}=m_{i-j+1}.
	\end{align*}
For any $f(z),g(z)\in \mr[[z]]$, let's define
	\begin{align*}
		\frac{d}{dt}\langle f(z),g(z)\rangle_t=\langle zf(z),g(z)\rangle_t=\langle f(z),z^{-1}g(z)\rangle_t.
	\end{align*}
	Under such an assumption, we have the following evolution formula for $(Q_n(z;t))^*$.

	\begin{proposition}
		The evolution of $(Q_n(z;t))^*$ in the positive flow satisfies
		\begin{align}\label{te22}
			\p_t (Q_n(z;t))^*=(Q_{n-1}(z;t))^*\eta_{n},\quad \eta_n=-\xi_{n-1}\psi_{n}^{-1},
		\end{align}
		where $\xi_n$ and $\psi_n$ are coefficients in three-term recurrence relation \eqref{newtr}.
	\end{proposition}
	\begin{proof}
		By taking the derivative to the equation
		\begin{align*}
			\langle z^i, Q_n(z;t)\rangle_t=0,\quad 0\leq i\leq n-1,
		\end{align*}
		we have
		\begin{align}\label{proppr}
			\langle z^i, \p_t Q_n(z;t)\rangle_t+\langle z^{i+1},Q_n(z;t)\rangle_t=0,\quad 0\leq i\leq n-1.
		\end{align}
		There are two explanations for this formula. On one hand, this formula means that
		\begin{align*}
			\langle z^{i+1},Q_n(z;t)\rangle_t=-\langle z^i,\p_tQ_n(z;t)\rangle_t,\quad 0\leq i\leq n-1.
		\end{align*}
		Since $\p_t (Q_n(z;t))^*$ could be expressed as a right linear combination of $(Q_0(z;t))^*,\,\cdots,\,(Q_{n-1}(z;t))^*$, and both sides are equal to zero when $i=0,\cdots,n-2$, it indicates that 
		\begin{align*}
			\p_t (Q_n(z;t))^*=-(Q_{n-1}(z;t))^*\eta_n, \quad \eta_n=-H_{n-1}^{-1}H_n,
		\end{align*}
		where $H_n$ is the normalization factor defined by \eqref{tip}.
		On the other hand, by taking the recurrence relation \eqref{newtr} into consideration, the left hand side of equation \eqref{proppr} implies that
		\begin{align*}
			\langle z^{i+1},Q_n(z;t)\rangle_t=\langle z^i,z^{-1}Q_n(z;t)\rangle_t=\langle z^i, Q_{n-1}(z;t)+\xi_{n-2}^*Q_{n-2}(z;t)-z^{-1}\psi_n^*Q_{n-1}(z;t)\rangle_t,
		\end{align*}
		and thus we get the relation
		\begin{align*}
			\langle z^{i+1},\psi_n^*Q_{n-1}(z;t)\rangle=\langle z^i,\p_tQ_n(z;t)+Q_{n-1}(z;t)+\xi_{n-2}^* Q_{n-2}(z;t)\rangle_t.
		\end{align*}
		If we take $i$ from $0$ to $n-2$, we immediately get that $H_{n-1}\psi_{n-1}=H_{n-2}\xi_{n-2}$, which means that $\eta_n=-\xi_{n-1}\psi_{n}^{-1}$. 
	\end{proof}	
	If we denote $D=\text{diag}(\eta_0,\eta_1,\cdots)$, then time evolution with respect to positive flow could be written into matrix form
	\begin{align*}
		\p_t R=R\mathbb{D},\quad \mathbb{D}=\Lambda D.
	\end{align*}
	This time evolution together with its spectral part \eqref{w1} result in the Lax integrability 
	\begin{align*}
		\p_t (\mathbb{B}\mathbb{A}^{-1})=[\mathbb{B}\mathbb{A}^{-1},\mathbb{D}].
	\end{align*}
	Moreover, the diagonal part and off-diagonal parts indicates another form of semi-discrete non-commutative relativistic Toda equation 
	\begin{align*}
		&\p_t \Psi-\Lambda^{-1}\p_t \Xi\Lambda=D-\Lambda D\Lambda^{-1},\\
		&\p_t\Xi\Lambda=-\Lambda D-(\Psi-\Lambda^{-1}\Xi\Lambda)^{-1}(D\Xi\Lambda-\Lambda D\Psi),
	\end{align*}
	which could be written in terms of $\xi_n$ and $\psi_n$ as
	\begin{align*}
		&\p_t \xi_n=\xi_n\psi_{n+1}^{-1}-\psi_n^{-1}\xi_n,\\
		&\p_t \psi_n=\xi_n\psi_{n+1}^{-1}-\psi_{n-1}^{-1}\xi_{n-1}.
	\end{align*}

	\section{non-commutative integrability}\label{sec_Poisson}
	In this section,
	we construct a non-commutative network  and use it to get the Poisson structure for the non-commutative leapfrog map and a family of involutive invariants.   
	
	\subsection{Preliminaries}
	
	Firstly,
	we recall some backgrounds and facts about non-commutative Poisson structures and networks needed in this section.
	They are taken from \cite{arthamonov18,arthamonov20,van2008double, gekhtman2012,gekhtman2016,Non} and occasionally slightly modified to fit our needs.
	
	\subsubsection{Non-commutative Poisson bracket}
	Let $\A$ be an associative algebra over a field $\mathbb{K}$.
	To define a non-commutative Poisson bracket,
	we firstly introduce the notion of a double bracket $\dl-,-\dr$.
	Then we compose this bracket with the algebra multiplication $\mu:\,\A\otimes \A\rightarrow \A$ to get a bilinear operation $\A\times \A\rightarrow \A$.
	By considering the cyclic space $\A^\natural:=\A/[\A,\A]$,
	we get a Lie bracket on $\A^\natural$,
	which provides an $H_0$-Poisson bracket \cite{van2008double}.
	
	\begin{definition}
		A double bracket on $\A$ is a bilinear map
		\begin{align*}
			\dl-,-\dr:\,\A\times \A\rightarrow \A\otimes \A,
		\end{align*}
		which satisfies
		\begin{align*}
			&\dl a,b\dr=-\dl b,a\dr^\tau,\\
			&\dl a,bc\dr=\dl a,b\dr\left(1\otimes c\right)
			+\left(b\otimes 1\right)
			\dl a,c\dr,
		\end{align*}
		where $(x\otimes y)^\tau:=y\otimes x$.
	\end{definition} 
	It can be deduced from these properties that 
	\begin{align*}
		\dl ab,c\dr
		=\left(1\otimes a\right)\dl b,c\dr
		+\dl a,c\dr\left(b\otimes 1\right).
	\end{align*}
	\begin{definition}
		A double bracket $\dl-,-\dr$ is called a \emph{double Poisson bracket} if it additionally satisfies a version of the Jacobi identity
		\begin{align*}
			\dl a,\dl b,c\dr\dr_L+
			\sigma\dl b,\dl c,a\dr\dr_L+
			\sigma^2\dl c,\dl a,b\dr\dr_L
			=0,
		\end{align*}
		where $\dl x, y\otimes z\dr_L:=\dl x,y\dr\otimes z$ and $\sigma$ is the permutation operator $\sigma:\,x\otimes y\otimes z\rightarrow z\otimes x\otimes y$.
	\end{definition}
	The relation between double and $H_0$-Poisson brackets are established by M. Van den Bergh \cite{van2008double} as follows.
	For a double bracket $\dl-,-\dr$ in $\A$,
	we introduce another operation $\{-,-\}:\A\times \A\rightarrow \A$
	such that 
	\begin{align*}
		\{a,b\}:=
		\mu(\dl a,b\dr),
	\end{align*}
	where $\mu:\A\otimes \A\rightarrow \A$ is a multiplication map.
	Moreover, if $\dl-,-\dr$ is a double Poisson bracket,
	then the induced bracket $\{-,-\}$ satisfies the properties
	\begin{align*}
		&\{a,b\}=-\{b,a\}\mod [\A,\A],\\
		&\{a,bc\}=\{a,b\}c+b\{a,c\},\\
		&\{a,\{b,c\}\}=
		\{\{a,b\},c\}+
		\{b,\{a,c\}\}.
	\end{align*}
	For the associate algebra $\A$,
	we can define the cyclic space, denoted $\A^\natural$,
	to be the vector space $\A/[\A,\A]$, by the linear span of all commutators.
	Then we define the bracket $\langle-,-\rangle:\,\A^\natural\times\A^\natural\to\A^\natural$ by
	\begin{align*}
		\langle a^\natural,b^\natural\rangle
		:=\{a,b\}^\natural,\quad
		\forall a,\,b\in \A.
	\end{align*}
	\begin{proposition}
		Suppose that $\dl-,-\dr$ is a double Poisson bracket.
		Then the induced bracket $\langle-,-\rangle$ is a Lie bracket,
		which provides an $H_0$-Poisson structure on $\A^\natural$.
	\end{proposition}

	Let ${\rm Rep}_M(\A):={\rm Hom}(\A,{\rm Mat}_M)/{\rm Ad}\, {\rm GL}_M$ be the representation space of $\A$.
	Since the trace map $\tr$ is well defined on elements of $\A/[\A,\A]$,
	there is a Poisson bracket on ${\rm Rep}_M(\A)$ and
	\begin{align*}
		\{\tr(a),\tr(b)\}
		=\tr\langle a^\natural,b^\natural\rangle.
	\end{align*}

	\subsubsection{Non-commutative network}
	Directed planar networks with weighted edges have been widely used in the study of totally non-negative matrices. A thorough investigation of Poisson geometry of directed networks was conducted in \cite{gekhtman2009,gekhtman2012}. Subsequently, Ovenhouse introduced the concept of non-commutative networks and studied its non-commutative Poisson structures \cite{Non}.
	The associated $r$-matrix formalism was discussed in \cite{arthamonov20}, with precise algebraic and conceptual results shown in \cite{Non}.
	
	A network $G=(V,E)$ is a directed planar graph with a vertex set $V$ and an edge set $E$.
	In this paper,
	we suppose that the network is drawn inside a cylinder.
	The vertex on the boundary is called a source if it has exactly one outcoming edge and no incoming edges,
	while the vertex is called a sink if the direction of the single edge reverses. 
	All internal vertices of $G$ have three adjacent edges.
	The internal vertex is called \emph{white} if it has exactly one incoming edge,
	or \emph{black} if it has exactly one outcoming edge (see \fig\ref{wb}).
	The weights of edges in $G$ take values in an associate algebra $\A$.
	A path in $G$ is a sequence $(v_1,e_1,v_2,e_2,\ldots,e_r,v_{r+1})$ of vertices and edges such that $e_i=(v_i,v_{i+1})$ for $i=1,\ldots,r$.
	The weight of such a path is defined as the product of the weight of each $e_i$ in order.
		\begin{figure}[htbp]
		\begin{tikzpicture}[line width=0.75pt]
			\draw [dashed] (0,0) circle[radius=1.2];
			\draw [dashed] (3,0) circle[radius=1.2];
			\draw [->](-1.2,0)--(-0.1,0);
			\draw (0,0) circle[radius=0.1];
			\draw [->](0.05,0.05*1.732)--(0.6,0.6*1.732);
			\draw [->](0.05,-0.05*1.732)--(0.6,-0.6*1.732);
			\node [below]at (-0.6,0){$x$};
			\node [left]at (0.35,0.6){$z$};
			\node [right]at (0.4,-0.5){$y$};
			\filldraw (3,0) circle[radius=0.1];
			\draw [->](2.4,0.6*1.732)--(2.95,0.05*1.732);
			\draw [->](2.4,-0.6*1.732)--(2.95,-0.05*1.732);
			\draw [->](3.1,0)--(4.2,0);
			\node[right]at (2.65,0.6){$c$};
			\node[left]at (2.6,-0.5){$b$};
			\node [below] at (3.6,0){$a$};
		\end{tikzpicture}
		\caption{Two types of internal vertices}
		\label{wb}
	\end{figure}
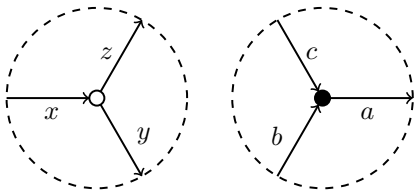 
	
	Next,
	we can choose an oriented curve $\rho$, called the \emph{cut},
	which connects two boundary components of the cylinder.
	The cylinder is cut into a planar graph, such as a rectangle.
	Suppose an oriented edge $\alpha$ has a weight $\wt(\alpha)$.
	If an dege intersecst $\rho$ at a point $i$,
	we define $\epsilon_i=1$, and $(\rho,\alpha)$ form an oriented basis on the plane. 
	Similarly, we define $\epsilon_i=-1$ if they have the opposite orientation. 
	The \emph{modified edge weight} of $\alpha$ is defined as
	\begin{align*}
		\mathrm{WT}(\alpha):=\wt(\alpha)\lambda^{\epsilon_i},
	\end{align*}
	where $\lambda\in\mathbb{K}$.
	Assuming that the graph $G$ is acyclic,
	we can define the boundary measurement matrix $\mathscr{B}(\lambda)=(b_{ij}(\lambda))$,
	where $b_{ij}(\lambda)$ is the sum of the weights of all the paths from the $i$-th source to the $j$-th sink.
	
	A natural operation on networks is their \emph{concatenation},
	which consists in gluing some sources (or sinks) of one network to some of the sinks (or sources) of the other.
	The weight of the new edge is defined as the product of the weights of the edges being glued.
	The boundary measurement matrix of the resulting network after concatenation is the product of the two boundary measurement matrices.
	Then we consider the Poisson structure on non-commutative networks.
	The Poisson structure associated with networks are expected to behave naturally under concatenation.
	This structure is related to the Poisson-Lie structure on groups \cite{gekhtman2012}.
	
	Firstly,
	let's consider local double brackets of weights of edges adjacent to a vertex.
	For any white vertex in \fig \ref{wb},
	local double brackets are defined by
	\begin{align}\label{pw}
		\dl x,y\dr=0,\quad
		\dl x,z\dr=0,\quad
		\dl y,z\dr=\frac{1}{2}\left(y\otimes z\right).
	\end{align}
	Similarly, for any black vertex, we have
	\begin{align}\label{pb}
		\dl a,b\dr=0,\quad
		\dl a,c\dr=0,\quad
		\dl b,c\dr=\frac{1}{2}\left(c\otimes b\right).
	\end{align}
	These local double brackets can be extended into a universal one in the network $G$ through concatenations. Moreover, their relations with Goldman bracket was discussed in \cite{Non}, in which a precise geometric interpretation was given for the double bracket (as well as its induced brackets). Another related structure is the quasi-Poisson structure \cite{massuyeau14}. In \cite{arthamonov18}, Arthamonov constructed a categorical version of such a quasi-Poisson structure, and proved the invariance of the double quasi-Poisson bracket under non-commutative mutations.

	Next we introduce some local transformations of networks, which do not change boundary measurements.
	The first type of move is called `gauge transformation',
	which change edge weights but not for the graph.
	Such a gauge transformation can be realized by the following procedure. For all incoming edges, we right multiply the weights by a parameter $t$, and for all outcoming edges, we left multiply the weights by $t^{-1}$. 
	One can easily check that this transformation keep the boundary measurements invariant.
	Besides,
	there are three types of Postnikov moves depicted in \fig \ref{pm}, which have the following description:
	\begin{enumerate}
		\item Perform the ``square move'' at each square face; 
		\item Perform the ``white-swap'' at each white-white edge;
		\item Perform the ``black-swap'' at each black-black edge.
	\end{enumerate}
	\begin{figure}[htbp]
		\centering
		\begin{tikzpicture}[line width=0.75pt]
			\node at (-1,0.5){$1.$};
			\draw (0,0)--(1-0.08,0);
			\draw (0,1)--(1-0.08,1);
			\draw (1,1-0.08)--(1,0.08);
			\draw [->](0,0)--(1/2,0);
			\draw [->](0,1)--(1/2,1);
			\draw (1+0.08,2-0.08);
			\draw [->](1,0+0.08)--(1,1/2);
			\draw (1,0) circle[radius=0.08];
			\filldraw (2,0) circle[radius=0.08];
			\filldraw (1,1) circle[radius=0.08];
			\draw(1.08,0)--(2,0);
			\draw[->](1.08,0)--(1.5,0);
			\draw(1,1)--(1.92,1);
			\draw(2,0)--(2,0.92);
			\draw [->](1.08,0)--(1.5,0);
			\draw [->](2,0.92)--(2,0.5);
			\draw (2,1)circle[radius=0.08];
			\draw (2,0)--(3,0);
			\draw (2.08,1)--(3,1);
			\draw [->](2,0)--(2.5,0);
			\draw [->](2.08,1)--(2.5,1);
			\draw [->](1,1)--(1.5,1);
			\node[left] at(1,0.5){$a$};
			\node[below] at(1.5,0){$b$};
			\node[right] at(2,0.5){$c$};
			\node[above] at(1.5,1){$d$};
		\end{tikzpicture}
		\hspace{9em}
		\begin{tikzpicture}[line width=0.75pt]
			\draw (0,0)--(1-0.08,0);
			\draw (0,1)--(1-0.08,1);
			\draw (1,1-0.08)--(1,0.08);
			\draw [->](0,0)--(1/2,0);
			\draw [->](0,1)--(1/2,1);
			\draw (1+0.08,2-0.08);
			\draw [->](1,1-0.08)--(1,1/2);
			\filldraw (1,0) circle[radius=0.08];
			\draw (2,0) circle[radius=0.08];
			\draw (1,1) circle[radius=0.08];
			\draw(1.08,0)--(1.92,0);
			\draw[->](1.08,0)--(1.5,0);
			\draw(1.08,1)--(1.92,1);
			\draw(2,0.08)--(2,0.92);
			\draw [->](1.08,0)--(1.5,0);
			\draw [->](2,0.08)--(2,0.5);
			\filldraw (2,1)circle[radius=0.08];
			\draw (2.08,0)--(3,0);
			\draw (2.08,1)--(3,1);
			\draw [->](2.08,0)--(2.5,0);
			\draw [->](2.08,1)--(2.5,1);
			\draw [->](1.08,1)--(1.5,1);
			\node[left] at(1,0.5){$dcf^{-1}$};
			\node[below] at(1.5,-0.1){$f$};
			\node[right] at(2,0.5){$f^{-1}ad$};
			\node[above] at(1.65,1.1){$dcf^{-1}bc^{-1}$};
		\end{tikzpicture}
		\newline
		\newline
		\begin{tikzpicture}[line width=0.75pt]
			\node at (-1,0.5){$2.$};
			\draw (0,0)--(0.92,0);
			\draw (1.08,0)--(1.92,0);
			\draw [->](2.08,0)--(3,0);
			\draw [->](1,0.08)--(1,1);
			\draw [->](2,-0.08)--(2,-1);
			\draw(1,0)circle[radius=0.08];
			\draw(2,0)circle[radius=0.08];
			\draw [->](0,0)--(0.5,0);
			\draw [->](1.08,0)--(1.5,0);
			\node[below]at (0.5,0){$a$};
			\node[below]at (1.5,0){$b$};
			\node[below]at (2.5,0){$c$};
			\node [left]at(1,0.7) {$x$};
			\node [right]at(2,-0.7) {$y$};
		\end{tikzpicture}
		\hspace{9em}
		\begin{tikzpicture}[line width=0.75pt]
			\draw (0,0)--(0.92,0);
			\draw (1.08,0)--(1.92,0);
			\draw [->](2.08,0)--(3,0);
			\draw [->](2,0.08)--(1,1);
			\draw [->](1,-0.08)--(2,-1);
			\draw(1,0)circle[radius=0.08];
			\draw(2,0)circle[radius=0.08];
			\draw [->](0,0)--(0.5,0);
			\draw [->](1.08,0)--(1.5,0);
			\node[below]at (0.5,0){$a$};
			\node[below]at (1.5,0){$b$};
			\node[below]at (2.5,0){$c$};
			\node [right]at(1.5,0.7) {$b^{-1}x$};
			\node [left]at(1.5,-0.7) {$by$};
		\end{tikzpicture}
		\newline
		\newline
		\begin{tikzpicture}[line width=0.75pt]
			\node at (-1,0.5){$3.$};
			\draw (0,0)--(0.92,0);
			\draw (1.08,0)--(1.92,0);
			\draw [->](2.08,0)--(3,0);
			\draw (1,0.08)--(1,1);
			\draw [->](1,1)--(1,0.5);
			\draw (2,-0.08)--(2,-1);
			\draw [->](2,-1)--(2,-0.5);
			\filldraw(1,0)circle[radius=0.08];
			\filldraw(2,0)circle[radius=0.08];
			\draw [->](0,0)--(0.5,0);
			\draw [->](1.08,0)--(1.5,0);
			\node[below]at (0.5,0){$a$};
			\node[below]at (1.5,0){$b$};
			\node[below]at (2.5,0){$c$};
			\node [left]at(1,0.7) {$x$};
			\node [right]at(2,-0.7) {$y$};
		\end{tikzpicture}
		\hspace{9em}
		\begin{tikzpicture}[line width=0.75pt]
			\draw (0,0)--(0.92,0);
			\draw (1.08,0)--(1.92,0);
			\draw [->](2.08,0)--(3,0);
			\draw (2,0.08)--(1,1);
			\draw [->](1,1)--(1.5,0.54);
			\draw (1,-0.08)--(2,-1);
			\draw [->](2,-1)--(1.5,-0.54);
			\filldraw(1,0)circle[radius=0.08];
			\filldraw(2,0)circle[radius=0.08];
			\draw [->](0,0)--(0.5,0);
			\draw [->](1.08,0)--(1.5,0);
			\node[below]at (0.5,0){$a$};
			\node[below]at (1.5,0){$b$};
			\node[below]at (2.5,0){$c$};
			\node [right]at(1.5,0.7) {$xb$};
			\node [left]at(1.5,-0.7) {$yb^{-1}$};
		\end{tikzpicture}
		\newline
		\caption{non-commutative Postnikov moves}
		\label{pm}
	\end{figure}
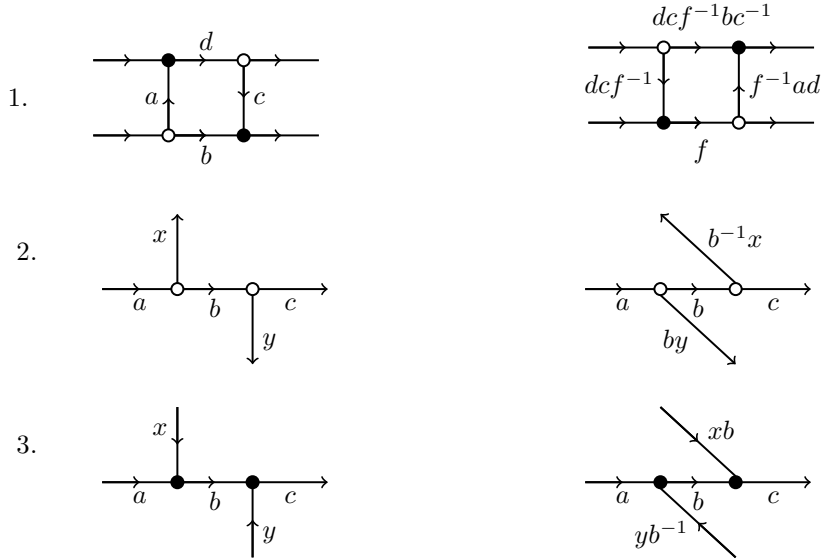 
	In the first type of Postnikov move,
	we can define a new weight $f:=b+adc$. 
	In this case, the associate algebra $\A$ can be extended to the \emph{free skew field} $\mr$ which consists of non-commutative rational expressions in a set of indeterminates, such as $f^{-1}$.
	The double bracket over $\A$ can be extended uniquely into this free skew field $\mr$ using the formulas
	\begin{align}\label{eq_inverse}
		\begin{split}
			\dl b,a^{-1}\dr=
			-\left(a^{-1}\otimes 1\right)
			\dl b,a\dr\left(1\otimes a^{-1}\right),\\
			\dl a^{-1},b\dr=
			-\left(1\otimes a^{-1}\right)
			\dl a,b\dr\left(a^{-1}\otimes 1\right).
		\end{split}
	\end{align}
	
	Recently,
	Gekhtman et al. made the use of Poisson geometry of directed networks on surfaces to generalize the pentagram map and the associated cluster structures in \cite{gekhtman2016}.
	Especially, they introduced the commutative leapfrog map and provided its Poisson and cluster structures by using networks.
	In next subsection, we extend the network of the leapfrog map to the non-commutative case.

	\subsection{Poisson structure of the non-commutative leapfrog map}

	In this part, let's consider a network in a cylinder which contains $N$ squares, see \fig \ref{net0}.
	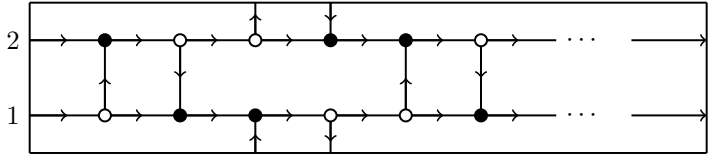
\begin{figure}[htbp]
		\begin{tikzpicture}[line width=0.75pt]
			\draw (0,0)--(1-0.08,0);
			\draw (0,1)--(1-0.08,1);
			\draw (1,1-0.08)--(1,0.08);
			\draw [->](0,0)--(1/2,0);
			\draw [->](0,1)--(1/2,1);
			\draw (1+0.08,2-0.08);
			\draw [->](1,0+0.08)--(1,1/2);
			\draw (1,0) circle[radius=0.08];
			\filldraw (2,0) circle[radius=0.08];
			\filldraw (1,1) circle[radius=0.08];
			\draw (2,1) circle[radius=0.08];
			\draw (2,1-0.08)--(2,0);
			\draw [->](1+0.08,0)--(3/2,0);
			\draw [->](1,1)--(3/2,1);
			\draw (1,1)--(2-0.08,1);
			\draw [->](2,1-0.08)--(2,1/2);
			\draw (1.08,0)--(2-0.08,0);
			\draw (2,0)--(4-0.08,0);
			\filldraw (3,0) circle[radius=0.08];
			\draw (4,0) circle[radius=0.08];
			\draw (4.08,0)--(5-0.08,0);
			\draw [->] (2,0)--(2.5,0);
			\draw [->] (3.08,0)--(3.5,0);
			\draw (3,0)--(3,-0.5);
			\draw [->] (3,-0.5)--(3,-0.25);
			\draw [->](4,-0.08)--(4,-0.35);
			\draw (4,-0.5)--(4,-0.08);
			\draw (2.08,1)--(3-0.08,1);
			\draw (3,1) circle[radius=0.08];
			\draw (3.08,1)--(4,1);
			\filldraw (4,1) circle[radius=0.08];
			\draw (4,1)--(5-0.08,1);
			\draw (3,1.08)--(3,1.5);
			\draw [->](3,1.08)--(3,1.35);
			\draw (4,1.08)--(4,1.5);
			\draw [->](4,1.5)--(4,1.25);
			\draw [->] (2.08,1)--(2.5,1);
			\draw [->](3.08,1)--(3.5,1);
			\draw [->](4.08,0)--(4.5,0);
			\draw [->](4.08,1)--(4.5,1);
			\node [left] at (0,0) {$1$};
			\node [left] at (0,1) {$2$};
			\node [right] at (7,0) {$\cdots$};
			\node [right] at (7,1) {$\cdots$};
			\draw (5,0.08)--(5,1-0.08);
			\draw (5.08,0)--(7,0);
			\draw (5,0) circle[radius=0.08];
			\filldraw (5,1) circle[radius=0.08];
			\draw [->](5,0.08)--(5,0.5);
			\draw (5,1)--(6-0.08,1);
			\draw [->](5.08,1)--(5.5,1);
			\draw (6.08,1)--(7,1);
			\draw (6,1) circle[radius=0.08];
			\draw (6,1-0.08)--(6,0.08);
			\filldraw (6,0)circle[radius=0.08];
			\draw [->](5.08,0)--(5.5,0);
			\draw [->](6.08,0)--(6.5,0);
			\draw [->](6.08,1)--(6.5,1);
			\draw [->](6,1-0.08)--(6,0.5);
			\draw (0,-0.5)--(0,1.5);
			\draw (0,-0.5)--(9,-0.5);
			\draw (0,1.5)--(9,1.5);
			\draw (9,-0.5)--(9,1.5);
			\draw [->](8,0)--(9,0);
			\draw [->](8,1)--(9,1);
		\end{tikzpicture}
		\caption{The network on a cylinder}\label{net0}
	\end{figure}
	\newline
	In this figure,
	the cylinder is cut along $\rho$ into a rectangle,
	so $\rho$ is both the top and bottom edge of the rectangle.
	Note that there are two sources and two sinks on the left and right boundaries respectively,
	so we can glue them together to obtain a network on the torus.
		
	Since the cylinder is composed by squares, let's look at each square and weights therein. 	
	For an elementary square in the left of \fig \ref{gauge},
	let's assume the associated edge weights take values in a skew field $\mr$.
		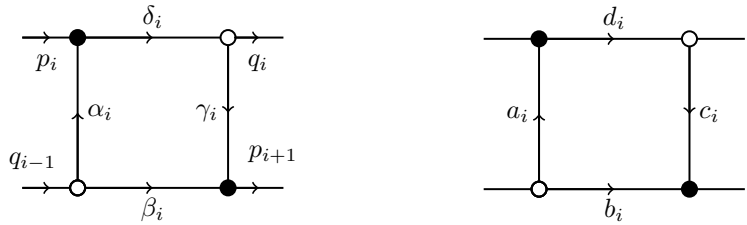
\begin{figure}[htbp]
		\begin{tikzpicture}[line width=0.75pt]
			\draw (-1.736,-1)--(-1.1,-1);
			\draw (-0.9,-1)--(1.736,-1);
			\draw (-1,-1) circle[radius=0.1];
			\filldraw (-1,1) circle[radius=0.1];
			\draw (-1,-1) circle[radius=0.1];
			\filldraw (1,-1) circle[radius=0.1];
			\draw (-1,-1) circle[radius=0.1];
			\draw (1,1) circle[radius=0.1];
			\draw (-1.736,1)--(0.9,1);
			\draw (1.1,1)--(1.736,1);
			\draw (-1,-0.9)--(-1,1);
			\draw (1,-1)--(1,0.9);
			\node [right] at (-1,0) {$\alpha_i$};
			\node [left] at (1,0) {$\gamma_i$};
			\node [above] at (0,1) {$\delta_i$};
			\node [below] at (0,-1) {$\beta_i$};
			\node [below] at (-1.4,0.9) {$p_i$};
			\draw [->](-1.736,1)--(-1.368,1);
			\node [below] at (1.4,0.9) {$q_i$};
			\draw [->](1.08,1)--(1.408,1);
			\node [above] at (-1.6,-0.9) {$q_{i-1}$};
			\draw [->](-1.736,-1)--(-1.368,-1);
			\node [above] at (1.6,-0.8) {$p_{i+1}$};
			\draw [->](1.08,-1)--(1.408,-1);
			\draw[->] (-0.9,-1)--(0,-1);
			\draw[->] (-1,1)--(0,1);
			\draw[->] (-1,-0.9)--(-1,0);
			\draw[->] (1,0.9)--(1,0);
		\end{tikzpicture}
		\hspace{6em}
		\begin{tikzpicture}[line width=0.75pt]
			\draw (-1.736,-1)--(-1.1,-1);
			\draw (-0.9,-1)--(1.736,-1);
			\draw (-1,-1) circle[radius=0.1];
			\filldraw (-1,1) circle[radius=0.1];
			\draw (-1,-1) circle[radius=0.1];
			\filldraw (1,-1) circle[radius=0.1];
			\draw (-1,-1) circle[radius=0.1];
			\draw (1,1) circle[radius=0.1];
			\draw (-1.736,1)--(0.9,1);
			\draw (1.1,1)--(1.736,1);
			\draw (-1,-0.9)--(-1,1);
			\draw (1,-1)--(1,0.9);
			\node [left] at (-1,0) {$a_i$};
			\node [right] at (1,0) {$c_i$};
			\node [above] at (0,1) {$d_i$};
			\node [below] at (0,-1) {$b_i$};
			\draw[->] (-0.9,-1)--(0,-1);
			\draw[->] (-1,1)--(0,1);
			\draw[->] (-1,-0.9)--(-1,0);
			\draw[->] (1,0.9)--(1,0);
		\end{tikzpicture}
		\caption{The edge weights of a square after a gauge transformation}
		\label{gauge}
	\end{figure} 
	\newline
	According to equations \eqref{pw} and \eqref{pb} together the properties of double brackets,
	we have
	\begin{align*}
		&\dl \alpha_i,p_i\dr=\frac{1}{2}p_i\otimes \alpha_i,\quad
		\dl\beta_i,\alpha_i\dr=\frac{1}{2}\beta_i\otimes\alpha_i,\\
		&\dl \gamma_i,q_i\dr=\frac{1}{2}\gamma_i\otimes q_i,\quad\,\,
		\dl\beta_i,\gamma_i\dr=\frac{1}{2}\gamma_i\otimes \beta_i,
	\end{align*}
while other brackets not shown above are zero.
	Then we could consider a gauge transformation at the corners of each square face,
	so that the edge weights of the $i$-th square are $a_i$, $b_i$, $c_i$, $d_i\in \mr$,
	while the others edge weights are set to $1$, which is the unity of the algebra $\mr$ (c.f. the right of \fig \ref{gauge}). These new weights are given by
	\begin{align*}
		a_i=\alpha_ip_i^{-1},\quad
		b_i=\beta_i,\quad
		c_i=q_i^{-1}\gamma_i,\quad
		d_i=p_i\delta_iq_i.
	\end{align*}
	Moreover, double brackets induced on variables after the gauge transformation are
	\begin{align}\label{poisson_abcd}
		\begin{split}
			&\dl b_i,a_i\dr=\frac{1}{2}b_i\otimes a_i,\quad\quad
			\dl b_i,c_i\dr=\frac{1}{2}c_i\otimes b_i,\\
			&\dl a_i,d_i\dr=\frac{1}{2}1\otimes a_id_i,\quad
			\dl c_i,d_i\dr=\frac{1}{2}d_ic_i\otimes 1.
		\end{split}
	\end{align}	
	The resulted network is shown in \fig \ref{net1}.
	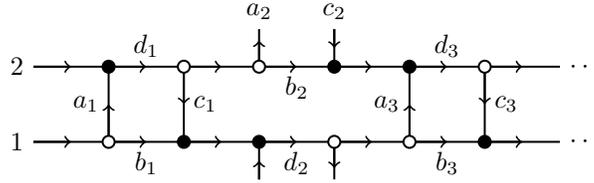
\begin{figure}[htbp]
		\begin{tikzpicture}[line width=0.75pt]
			\draw (0,0)--(1-0.08,0);
			\draw (0,1)--(1-0.08,1);
			\draw (1,1-0.08)--(1,0.08);
			\draw [->](0,0)--(1/2,0);
			\draw [->](0,1)--(1/2,1);
			\draw (1+0.08,2-0.08);
			\draw [->](1,0+0.08)--(1,1/2);
			\draw (1,0) circle[radius=0.08];
			\filldraw (2,0) circle[radius=0.08];
			\filldraw (1,1) circle[radius=0.08];
			\draw (2,1) circle[radius=0.08];
			\draw (2,1-0.08)--(2,0);
			\draw [->](1+0.08,0)--(3/2,0);
			\draw [->](1,1)--(3/2,1);
			\draw (1,1)--(2-0.08,1);
			\draw [->](2,1-0.08)--(2,1/2);
			\draw (1.08,0)--(2-0.08,0);
			\draw (2,0)--(4-0.08,0);
			\filldraw (3,0) circle[radius=0.08];
			\draw (4,0) circle[radius=0.08];
			\draw (4.08,0)--(5-0.08,0);
			\draw [->] (2,0)--(2.5,0);
			\draw [->] (3.08,0)--(3.5,0);
			\draw (3,0)--(3,-0.5);
			\draw [->] (3,-0.5)--(3,-0.25);
			\draw [->](4,-0.08)--(4,-0.35);
			\draw (4,-0.5)--(4,-0.08);
			\draw (2.08,1)--(3-0.08,1);
			\draw (3,1) circle[radius=0.08];
			\draw (3.08,1)--(4,1);
			\filldraw (4,1) circle[radius=0.08];
			\draw (4,1)--(5-0.08,1);
			\draw (3,1.08)--(3,1.5);
			\draw [->](3,1.08)--(3,1.35);
			\draw (4,1.08)--(4,1.5);
			\draw [->](4,1.5)--(4,1.25);
			\draw [->] (2.08,1)--(2.5,1);
			\draw [->](3.08,1)--(3.5,1);
			\draw [->](4.08,0)--(4.5,0);
			\draw [->](4.08,1)--(4.5,1);
			\node [below] at (1.5,0) {$b_1$};
			\node [below] at (3.5,0) {$d_2$};
			\node [below] at (5.5,0) {$b_3$};
			\node [left] at (1,0.5) {$a_1$};
			\node [left] at (5,0.5) {$a_3$};
			\node [right] at (2,0.5) {$c_1$};
			\node [right] at (6,0.5) {$c_3$};
			\node [above] at (1.5,1) {$d_1$};
			\node [below] at (3.5,1) {$b_2$};
			\node [above] at (5.5,1) {$d_3$};
			\node [above] at (3,1.5) {$a_2$};
			\node [above] at (4,1.5) {$c_2$};
			\node [left] at (0,0) {$1$};
			\node [left] at (0,1) {$2$};
			\node [right] at (7,0) {$\cdots$};
			\node [right] at (7,1) {$\cdots$};
			\draw (5,0.08)--(5,1-0.08);
			\draw (5.08,0)--(7,0);
			\draw (5,0) circle[radius=0.08];
			\filldraw (5,1) circle[radius=0.08];
			\draw [->](5,0.08)--(5,0.5);
			\draw (5,1)--(6-0.08,1);
			\draw [->](5.08,1)--(5.5,1);
			\draw (6.08,1)--(7,1);
			\draw (6,1) circle[radius=0.08];
			\draw (6,1-0.08)--(6,0.08);
			\filldraw (6,0)circle[radius=0.08];
			\draw [->](5.08,0)--(5.5,0);
			\draw [->](6.08,0)--(6.5,0);
			\draw [->](6.08,1)--(6.5,1);
			\draw [->](6,1-0.08)--(6,0.5);
		\end{tikzpicture}
		\caption{The network after gauge transformations}\label{net1}
	\end{figure}
	
	To simplify this network, we perform another gauge transformation,
	so that all weights become $1$ except those on bottom and left edges of every square face and the last two edges in the right. 
	The result is depicted in \fig \ref{net2}.
	\begin{figure}[htbp]
		\begin{tikzpicture}[line width=0.75pt]
			\draw (0,0)--(1-0.08,0);
			\draw (0,1)--(1-0.08,1);
			\draw (1,1-0.08)--(1,0.08);
			\draw [->](0,0)--(1/2,0);
			\draw [->](0,1)--(1/2,1);
			\draw (1+0.08,2-0.08);
			\draw [->](1,0+0.08)--(1,1/2);
			\draw (1,0) circle[radius=0.08];
			\filldraw (2,0) circle[radius=0.08];
			\filldraw (1,1) circle[radius=0.08];
			\draw (2,1) circle[radius=0.08];
			\draw (2,1-0.08)--(2,0);
			\draw [->](1+0.08,0)--(3/2,0);
			\draw [->](1,1)--(3/2,1);
			\draw (1,1)--(2-0.08,1);
			\draw [->](2,1-0.08)--(2,1/2);
			\draw (1.08,0)--(2-0.08,0);
			\draw (2,0)--(4-0.08,0);
			\filldraw (3,0) circle[radius=0.08];
			\draw (4,0) circle[radius=0.08];
			\draw (4.08,0)--(5-0.08,0);
			\draw [->] (2,0)--(2.5,0);
			\draw [->] (3.08,0)--(3.5,0);
			\draw (3,0)--(3,-0.5);
			\draw [->] (3,-0.5)--(3,-0.25);
			\draw [->](4,-0.08)--(4,-0.35);
			\draw (4,-0.5)--(4,-0.08);
			\draw (2.08,1)--(3-0.08,1);
			\draw (3,1) circle[radius=0.08];
			\draw (3.08,1)--(4,1);
			\filldraw (4,1) circle[radius=0.08];
			\draw (4,1)--(5-0.08,1);
			\draw (3,1.08)--(3,1.5);
			\draw [->](3,1.08)--(3,1.35);
			\draw (4,1.08)--(4,1.5);
			\draw [->](4,1.5)--(4,1.25);
			\draw [->] (2.08,1)--(2.5,1);
			\draw [->](3.08,1)--(3.5,1);
			\draw [->](4.08,0)--(4.5,0);
			\draw [->](4.08,1)--(4.5,1);
			\node [below] at (1.5,0) {$Y_1$};
			\node [left] at (1,0.5) {$X_1$};
			\node [below] at (3.5,1) {$Y_2$};
			\node [above] at (3,1.5) {$X_2$};
			\node [below] at (5.5,0) {$Y_3$};
			\node [left] at (5,0.5) {$X_3$};
			\node [left] at (0,0) {$1$};
			\node [left] at (0,1) {$2$};
			\node [right] at (7,0) {$\cdots$};
			\node [right] at (7,1) {$\cdots$};
			\draw (5,0.08)--(5,1-0.08);
			\draw (5.08,0)--(7,0);
			\draw (5,0) circle[radius=0.08];
			\filldraw (5,1) circle[radius=0.08];
			\draw [->](5,0.08)--(5,0.5);
			\draw (5,1)--(6-0.08,1);
			\draw [->](5.08,1)--(5.5,1);
			\draw (6.08,1)--(7,1);
			\draw (6,1) circle[radius=0.08];
			\draw (6,1-0.08)--(6,0.08);
			\filldraw (6,0)circle[radius=0.08];
			\draw [->](5.08,0)--(5.5,0);
			\draw [->](6.08,0)--(6.5,0);
			\draw [->](6.08,1)--(6.5,1);
			\draw [->](6,1-0.08)--(6,0.5);
			\draw (8,0)--(9,0);
			\draw [->](8,0)--(8.5,0);
			\draw (8,1)--(9,1);
			\draw [->](8,1)--(8.5,1);
			\node [above] at(8.5,1){$Z$};
			\node [below] at(8.5,0){$Z$};
		\end{tikzpicture}
		\caption{The new edge weights}
		\label{net2}
	\end{figure}
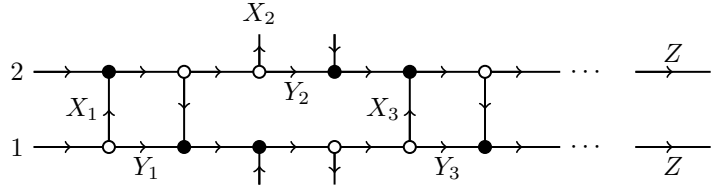
	\newline
 According to the rule of gauge transformation again, edge weights in \fig \ref{net2} are given by
	\begin{align}\label{xy2}
		\begin{split}
			X_i&=z_{i-1}a_ic_{i-1}^{-1}z_{i-1}^{-1},\quad
			i\geq 2,\\
			Y_i&=z_{i-1}b_ic_i^{-1}d_i^{-1}c_{i-1}^{-1}z_{i-1}^{-1},\quad
			i\geq 2,\\
			Z&=z_Nc_N,
		\end{split}
	\end{align}
	where $z_i=d_1c_1\cdots d_{i-1}c_{i-1}d_i$ for $i=1,2,\ldots,N$.
	Note that indices are not read cyclically with period $N$, but conjugate by $Z$, i.e. 
	\begin{align*}
		X_{i+N}=ZX_iZ^{-1},
	\end{align*}
	which gives us
	\begin{align*}
		X_1&=Z^{-1}X_{N+1}Z=Z^{-1}z_Na_1c_N^{-1}z_{N-1}^{-1}Z=c_N^{-1}a_1,\\
		Y_1&=Z^{-1}Y_{N+1}Z=Z^{-1}z_Nb_1c_1^{-1}d_1^{-1}c_N^{-1}z_{N-1}^{-1}Z=c_N^{-1}b_1c_1^{-1}d_1^{-1}.
	\end{align*}

	\begin{theorem}
		The $H_0$-Poisson brackets for the edge weights in \fig \ref{net2} are
		\begin{align}\label{poi}
			\begin{split}
				\langle Y_i,X_i\rangle=Y_iX_i,\quad
				\langle X_{i+1}, Y_i\rangle=X_{i+1}Y_i,\quad
				\langle Y_{i+1},Y_i\rangle=Y_{i+1}Y_i,
			\end{split}
		\end{align}
		with exceptions
		\begin{align}\label{poi_e}
			\langle X_1,Y_N\rangle=X_1Z^{-1}Y_NZ,\quad
			\langle Y_1,Y_N\rangle=Y_1Z^{-1}Y_NZ.
		\end{align}
	\end{theorem}
	\begin{proof}

		Since $a^\natural=(bab^{-1})^\natural$,
		we have $\langle a,c\rangle=\langle bab^{-1},c\rangle$.
		So we will compute the brackets of the conjugate elements
		\begin{align*}
			x_i=a_ic_{i-1}^{-1},\quad
			y_i=b_ic_i^{-1}d_i^{-1}c_{i-1}^{-1},\quad
			i\geq 2.
		\end{align*}
		For example,
		we have
		\begin{align*}
			\langle Y_j,X_i\rangle=\langle y_j,x_i\rangle.
		\end{align*}
		We need to calculate their double brackets
		\begin{align*}
			\dl y_j,x_i\dr
			=&\dl b_jc_j^{-1}d_j^{-1}c_{j-1}^{-1},a_ic_{i-1}^{-1} \dr\\
			=&
			\frac{1}{2}\delta_{i,j}
			\left( b_jc_j^{-1}d_j^{-1}c_{j-1}^{-1}\otimes a_ic_{i-1}^{-1}
			+
			a_ic_{i-1}^{-1}\otimes
			b_jc_j^{-1}d_j^{-1}c_{j-1}^{-1}  \right)\\
			&-\delta_{i-1,j}
			a_ic_{i-1}^{-1}d_j^{-1}c_{j-1}^{-1}\otimes
			b_jc_j^{-1},
		\end{align*}
		where we use the identities \eqref{eq_inverse}.
		Therefore, induced bracket is
		\begin{align*}
			\langle y_j,x_i\rangle
			=\delta_{i,j}y_jx_i-\delta_{i-1,j}x_id_j^{-1}c_{j-1}^{-1}y_jc_{j-1}d_j,
		\end{align*}
		which yields
		\begin{align*}
			\langle Y_j,X_i\rangle
			=\delta_{i,j}Y_jX_i
			-\delta_{i-1,j}X_{i}Y_j.
		\end{align*}
		Calculations of $\langle Y_i,Y_j\rangle$ and $\langle X_i,X_j\rangle$ are similar.
		
		For the exceptions,
		we have
		\begin{align*}
			\langle X_1,Y_N\rangle=\langle c_N^{-1}a_1,b_{N}c_N^{-1}d_N^{-1}c_{N-1}^{-1}\rangle,
		\end{align*}
		whose corresponding double bracket is
		\begin{align*}
			\dl c_N^{-1}a_1,b_{N}c_N^{-1}d_N^{-1}c_{N-1}^{-1}\dr
			=\frac{1}{2}d_Na_1^{-1}\otimes
			c_Na_1^{-1}b_N^{-1}a_N^{-1}
			+\frac{1}{2}d_Na_1^{-1}c_Na_1^{-1}\otimes
			b_N^{-1}a_N^{-1},
		\end{align*}
		which yields the induced bracket
		\begin{align*}
			\langle X_1,Y_N\rangle
			=X_1Z^{-1}Y_NZ.
		\end{align*}
		The second equation in \eqref{poi_e} could be similarly verified.
	\end{proof}
	
	Well prepared, we consider Postnikov moves on the non-commutative network over a cylinder,
	which can be treated as a discrete time evolution in the non-commutative leapfrog map.
	Taking Postnikov moves in \fig \ref{pm},
	the new edge weights are given by
	\begin{align}\label{ab_tr}
		\tilde{a}_i=d_ic_if_i^{-1},\quad
		\tilde{b}_i=f_i,\quad
		\tilde{c}_i=f_i^{-1}a_id_i,\quad
		\tilde{d}_i=d_ic_if_i^{-1}b_ic_i^{-1},
	\end{align}
	where $f_i=b_i+a_id_ic_i$.
	The results are pictured in Fig \ref{edge}.
	\begin{figure}[htbp]
		\begin{tikzpicture}[line width=0.75pt]
			\draw [dashed] (0,0) circle[radius=2];
			\draw (-1.736,-1)--(-1.1,-1);
			\draw (-0.9,-1)--(1.736,-1);
			\draw (-1,-1) circle[radius=0.1];
			\filldraw (-1,1) circle[radius=0.1];
			\draw (-1,-1) circle[radius=0.1];
			\filldraw (1,-1) circle[radius=0.1];
			\draw (-1,-1) circle[radius=0.1];
			\draw (1,1) circle[radius=0.1];
			\draw (-1.736,1)--(0.9,1);
			\draw (1.1,1)--(1.736,1);
			\draw (-1,-0.9)--(-1,1);
			\draw (1,-1)--(1,0.9);
			\node [left] at (-1,0) {$\tilde{c}_{i-1}$};
			\node [right] at (1,0) {$\tilde{a}_{i+1}$};
			\node [above] at (0,1) {$\tilde{b}_{i}$};
			\node [below] at (0,-1) {$\tilde{d}_{i}$};
			\draw[->] (-0.9,-1)--(0,-1);
			\draw[->] (-1,1)--(0,1);
			\draw[->] (-1,-0.9)--(-1,0);
			\draw[->] (1,0.9)--(1,0);
		\end{tikzpicture}
		\caption{Edge weights after Postnikov moves}
		\label{edge}
	\end{figure}
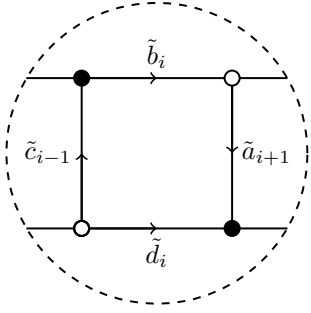 
	\newline
	For edge weights after Postnikov moves, they satisfy the following double brackets
	\begin{align}\label{poisson_abcd1}
		\begin{split}
			& \dl\widetilde{b}_i, \widetilde{a}_i\dr=\frac{1}{2}(\widetilde{a}_i \widetilde{b}_i \otimes 1), \quad
			\dl\widetilde{a}_i, \widetilde{d}_i\dr=\frac{1}{2}(\widetilde{a}_i \otimes \widetilde{d}_i), \\
			& \dl\widetilde{b}_i, \widetilde{c}_i\dr=\frac{1}{2}(1 \otimes \widetilde{b}_i \widetilde{c}_i) ,
			\quad
			\dl\widetilde{c}_i, \widetilde{d}_i\dr=\frac{1}{2}(\widetilde{d}_i \otimes \widetilde{c}_i).
		\end{split}
	\end{align}
	Then we take a similar gauge transformation in \fig \ref{net2},
	so that we set as many edge weights as possible equal to $1$.
	The corresponding new weights are
	\begin{align}\label{tildexy}
		\tilde{X}_i=\xi_{i-1}\tilde{c}_{i-1}\tilde{a}_i^{-1}\xi_{i-1}^{-1},\quad
		\tilde{Y}_i=\xi_{i-1}\tilde{d}_{i}\tilde{a}_{i+1}^{-1}\tilde{b}_{i}^{-1}\tilde{a}_i^{-1}\xi_{i-1}^{-1},
	\end{align}
	where $\xi_i=\tilde{a}_1\tilde{b}_1\cdots \tilde{a}_i\tilde{b}_i$.
	According to \eqref{ab_tr},
	we have $\xi_N=Z=d_1c_1\cdots d_Nc_N$.
	\begin{proposition}	
		The Postnikov moves could be described as the following dynamics
		\begin{align}
			\begin{split}\label{xy_eq}
				\tilde{X}_i&=(X_{i-1}+Y_{i-1})^{-1}X_{i-1}(X_{i}+Y_{i}),\\
				\tilde{Y}_i&=(X_{i}+Y_{i})^{-1}Y_{i}(X_{i+1}+Y_{i+1}),\\
				\tilde{Z}&=Z,
			\end{split}
		\end{align}
		where $(X_i,Y_i,Z)$ are edge weights before the move, and $(\tilde{X}_i,\tilde{Y}_i,\tilde{Z})$ are edge weights after the move.
	\end{proposition}
	It should be noted that after imposing the $N$-twisted condition,
	equations \eqref{leap_a} and \eqref{leap_b} for the non-commutative leapfrog map 
	is the same as \eqref{xy_eq}. Moreover, we could show that the $H_0$-poisson bracket defined over the networks is invariant under the leapfrog map.
	\begin{theorem}
		The brackets \eqref{poi} and \eqref{poi_e} are invariant under the leapfrog map $T$.
	\end{theorem}
	\begin{proof}
		Firstly,
		we define two conjugate variables of $\tilde{X}_i$, $\tilde{Y}_i$ by
		\begin{align*}
			\tilde{x}_i=\xi_{i-1}^{-1}\tilde{X}_i\xi_{i-1}=\tilde{c}_{i-1}\tilde{a}_i^{-1},\quad
			\tilde{y}_i=\xi_{i-1}^{-1}\tilde{Y}_i\xi_{i-1}=\tilde{d}_{i}\tilde{a}_{i+1}^{-1}\tilde{b}_{i}^{-1}\tilde{a}_i^{-1}.
		\end{align*}
		Then we compute three possible double brackets for $\dl \tilde{x}_i,\tilde{x}_j\dr$, $\dl \tilde{x}_i,\tilde{y}_j\dr$ and $\dl \tilde{y}_i,\tilde{y}_j\dr$.
		For $\dl \tilde{x}_i,\tilde{x}_j\dr$,
		a simple calculation by using the Leibniz rules and \eqref{poisson_abcd1} gives
		\begin{align*}
			\dl \tilde{x}_i,\tilde{x}_j\dr=0.
		\end{align*}
		The double brackets for $\tilde{y}_i$ and $\tilde{y}_j$ are
		\begin{align*}
			\dl \tilde{y}_i,\tilde{y}_j \dr
			=&-\frac{1}{2}\delta_{i+1,j}\left(
			\tilde{a}_{i+1}\tilde{d}_i^{-1}\tilde{y}_i
			\otimes
			\tilde{d}_i\tilde{a}_{i+1}^{-1}\tilde{y}_j
			+\tilde{y}_j\tilde{b}_i^{-1}\tilde{a}_i^{-1}\otimes
			\tilde{y}_i\tilde{a}_i\tilde{b}_i
			\right)\\
			&+\frac{1}{2}\delta_{i,j+1}\left(
			\tilde{d}_j\tilde{a}_{j+1}^{-1}\tilde{y}_i\otimes
			\tilde{a}_{j+1}\tilde{d}_j^{-1}\tilde{y}_j+
			\tilde{y}_j\tilde{a}_j\tilde{b}_j\otimes
			\tilde{y}_i\tilde{b}_j^{-1}\tilde{a}_j^{-1}
			\right),
		\end{align*}
		which induces
		\begin{align*}
			\langle
			\tilde{y}_i,\tilde{y}_j\rangle
			=-\delta_{i+1,j}\tilde{y}_i\tilde{a}_i\tilde{b}_i\tilde{y}_j\tilde{b}_i^{-1}\tilde{a}_i^{-1}
			+\delta_{i,j+1}\tilde{y}_j\tilde{a}_j\tilde{b}_j\tilde{y}_i\tilde{b}_j^{-1}\tilde{a}_j^{-1}.
		\end{align*}
		Therefore,
		we have
		\begin{align*}
			\langle \tilde{Y}_i,\tilde{Y}_j\rangle
			=\left(
			\delta_{i,j+1}-\delta_{i+1,j}\right)\tilde{Y}_i\tilde{Y}_j.
		\end{align*}
		The rest are similar, and thus we have checked that the bracket relations are invariant under the leapfrog map.
	\end{proof}
	
	Then we want to demonstrate the boundary measurement matrix and re-derive its Lax integrability. 
	Let's cut the cylinder as in Fig. \ref{cut}.
	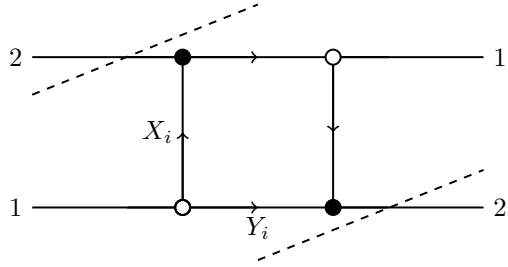
\begin{figure}[htbp]
		\begin{tikzpicture}[line width=0.75pt]
			\draw (-1.736,-1)--(-1.1,-1);
			\draw (-0.9,-1)--(1.736,-1);
			\draw (-1,-1) circle[radius=0.1];
			\filldraw (-1,1) circle[radius=0.1];
			\draw (-1,-1) circle[radius=0.1];
			\filldraw (1,-1) circle[radius=0.1];
			\draw (-1,-1) circle[radius=0.1];
			\draw (1,1) circle[radius=0.1];
			\draw (-1.736,1)--(0.9,1);
			\draw (1.1,1)--(1.736,1);
			\draw (-1,-0.9)--(-1,1);
			\draw (1,-1)--(1,0.9);
			\node [left] at (-1,0) {$X_i$};
			\node [below] at (0,-1) {$Y_i$};
			\draw[->] (-0.9,-1)--(0,-1);
			\draw[->] (-1,1)--(0,1);
			\draw[->] (-1,-0.9)--(-1,0);
			\draw[->] (1,0.9)--(1,0);
			\draw (-3,-1)--(-1.1,-1);
			\draw (-3,1)--(-1.1,1);
			\draw (3,-1)--(1.1,-1);
			\draw (3,1)--(1.1,1);
			\node [left] at (-3,-1) {$1$};
			\node [left] at (-3,1) {$2$};
			\node [right] at (3,-1) {$2$};
			\node [right] at (3,1) {$1$};
			\draw[dashed] (-3,0.5)--(0,1.7);
			\draw[dashed] (0,-1.7)--(3,-0.5);
		\end{tikzpicture}
		\caption{The cut in one  elementary network}
		\label{cut}
	\end{figure} 
	\newline The corresponding boundary measurement matrix of the $i$-th elementary network is 
	\begin{align*}
		\mathscr{B}_i(\mu)
		=\begin{pmatrix}
			X_i&\mu(X_i+Y_i)\\
			\mu&\mu^2
		\end{pmatrix}.
	\end{align*}
	The boundary measurement matrix of the entire network is the product 
	\begin{align*}
		\mathscr{B}(\mu):=\mathscr{B}_1(\mu)\cdots\mathscr{B}_N(\mu).
	\end{align*}
	Let $\tr(\mathscr{B}(\mu)^i)=\sum_j t_{i,j}\mu^j $,
	where $t_{i,j}$ can be written as non-commutative polynomials of $X_i$ and $Y_i$.
	According to \cite{Non},
	these functions $t_{i,j}$ form an involutive family with respect to this induced Poisson structure.
	Finally,
	we give a Lax matrix for the leapfrog map.
	\begin{proposition}
		The boundary measurement matrix $\mathscr{B}_i$ can be written as
		\begin{align*}
			\mathscr{B}_i=A_iL_{i}A_{i+1}^{-1},
		\end{align*}
		where
		\begin{align*}
			L_i=\begin{pmatrix}
				0&\lambda Y_i\\
				1&\lambda+X_{i+1}
			\end{pmatrix},
			\quad
			A_i=\begin{pmatrix}
				\lambda^{-1}&\lambda^{-1}X_i\\
				0&\mu^{-1}
			\end{pmatrix},\quad
			\lambda:=\mu^2.
		\end{align*}
		
	\end{proposition}
	Therefore,
	$\mathscr{B}(\lambda)$ is conjugate to the matrix $L(\lambda):=L_1(\lambda)\cdots L_N(\lambda)$.	
	Let $\Phi_i=(\phi_{i},\psi_{i})$ and assume that it satisfies the linear equation
	\begin{align*}
		\Phi_{i+1}=\Phi_iL_i,
	\end{align*}
	which is equivalent to
	\begin{align*}
		\psi_{i}=\phi_{i+1},\quad
		\phi_{i+2}-X_{i+1}\phi_{i+1}=\lambda\left( \phi_{i+1}+Y_i\phi\right).
	\end{align*}
	Therefore $L_i(\lambda)$ can be treated the Lax matrix of the leapfrog map,
	while the corresponding discrete time evoluton is
	\begin{align*}
		T\phi_i=\phi_{i+1}+Y_i\phi.
	\end{align*}
	According to this proposition above,
	$\mathscr{B}(\lambda)$ and $L(\lambda)$ possess the same spectral invariants.
	
	\begin{remark}
		
		To get the complete integrability,
		it is crucial to obtain a maximal family of involutive invariants.
		However, 
		it remains an open question about whether theses invariants are indeed maximal.
	\end{remark}
	
	\section*{Acknowledgement}
	The authors would like to thank Professor Vladimir Retakh for his comments. BW was supported by the National Natural Science Foundation of China (Nos. 12201325 \& 12235007). SHL was supported by the National Natural Science Foundation of China (Nos. 12101432 \& 12175155).

	\appendix\label{appendixa}
	\section*{Appendix A}
	\renewcommand{\thesection}{A} 
	\setcounter{equation}{0}
	\setcounter{theorem}{0}
	\subsection{Basic quasi-determinant identities}\label{AppendixA1}
	
	In the appendix, we list some basic facts that are used in the article. Firstly, we give the definition of a quasi-determinant.
	\begin{definition}\label{qd-def}
		Let $A$ be an $n\times n$ matrix over an associative algebra $\mathcal{R}$.
		For $i, j=1, 2, \dots, n$, let $r_i^j$ be the $i$-th row of $A$ without the $j$-th entry, $c_j^i$ be the $j$-th column without the $i$-th entry, and $A^{i,j}$ be the submatrix of $A$ without the $i$-th row and $j$-th column of $A$.
		Assume that $A^{i,j}$ is invertible.
		Then there are $n^2$ quasi-determinants of $A$, denoted as $|A|_{i,j}$ for $1\leq i, j\leq n$, as follows
		\begin{align*}
			|A|_{i,j}=a_{i,j}-r_i^j\left(A^{i,j}\right)^{-1}c_j^i,
		\end{align*}
		where $a_{i,j}$ is the $(i,j)$-th entry of $A$. For convenience, we denote the quasi-determinant expanded by $(i,j)$-position as
		\begin{align*}
			|A|_{i,j}=\left|\begin{array}{cc}
				A^{i,j}&c_j^i\\
				r_i^j&\boxed{a_{ij}}
			\end{array}
			\right|.
		\end{align*}
	\end{definition}
	As is known, cross-ratio is defined by a linear system of coordinates. To get the definition of non-commutative cross-ratio, we need to solve a linear system with non-commutative coefficients. Thus, 
	the following is to show how to make use of quasi-determinants to solve this problem.
	
	\begin{proposition}(\cite[Thm. 1.6.1]{gelfand05})\label{p-ls}
		Let $A=(a_{i,j})$ be an $n\times n$ matrix over an associate algebra $\mathcal{R}$.
		Assume that all the quasi-determinants $|A|_{i,j}$  are defined and invertible.
		Then
		\begin{align*}
			\left\{\begin{array}{c}
				a_{1,1}x_1+\cdots+a_{1,n}x_n=\xi_1\\
				\vdots\\
				a_{n,1}x_1+\cdots+a_{n,n}x_n=\xi_n\end{array}
			\right.
		\end{align*}
		has a solution $x_i\in\mathcal{R}$ if and only if
		\begin{align*}
			x_i=\sum_{j=1}^n|A|_{j,i}^{-1}\xi_j.
		\end{align*}
	\end{proposition}
	
	The following is an equivalent condition for a zero quasi-determinant, which is used in giving constraints for points in $S$ and $S^-$.
	\begin{proposition}(\cite[Prop. 1.4.6]{gelfand05})\label{p-equi}
		The following statements are equivalent if the quasi-determinant $|A|_{ij}$ is defined.\\
		({\romannumeral1})~~$|A|_{ij}$=0;\\
		({\romannumeral2})~~The $i$-th row of the matrix $A$ is a left linear combination of the other rows of $A$.\\
		({\romannumeral3})~~The $j$-th column of the matrix $A$ is a right linear combination of the other columns of $A$.
	\end{proposition}
	
	Moreover, we have the following non-commutative Jacobi identity \cite{gilson07}
	\begin{align}\label{ncj1}
		\left|\begin{array}{ccc}
			A&B&C\\
			D&f&g\\
			E&h&\boxed{i}
		\end{array}
		\right|
		=\left|
		\begin{array}{cc}
			A&C\\
			E&\boxed{i}
		\end{array}
		\right|-\left|\begin{array}{cc}
			A&B\\
			E&\boxed{h}
		\end{array}
		\right|\left|\begin{array}{cc}
			A&B\\
			D&\boxed{f}
		\end{array}
		\right|^{-1}\left|\begin{array}{cc}
			A&C\\
			D&\boxed{g}
		\end{array}
		\right|.
	\end{align}
	which could be viewed as a special case of homological relation in terms of quasi-Pl\"ucker coordinates
	\begin{align}
		\left|\begin{array}{ccc}
			A&B&C\\
			D&f&g\\
			E&\boxed{h}&i\end{array}
		\right|&=\left|\begin{array}{ccc}
			A&B&C\\
			D&f&g\\
			E&h&\boxed{i}\end{array}
		\right|\left|\begin{array}{ccc}
			A&B&C\\
			D&f&g\\
			0&\boxed{0}&1\end{array}
		\right|,\label{hm1}\\
		\left|\begin{array}{ccc}
			A&B&C\\
			D&f&\boxed{g}\\
			E&h&i\end{array}
		\right|&=\left|\begin{array}{ccc}
			A&B&0\\
			D&f&\boxed{0}\\
			E&h&1\end{array}
		\right|\left|\begin{array}{ccc}
			A&B&C\\
			D&f&g\\
			E&h&\boxed{i}\end{array}
		\right|,\label{hm2}
	\end{align}
	which have been used in this article.


\begin{thebibliography}{1}
		
		\bibitem{arthamonov18}
		S. Arthamonov.
		\newblock Generalized quasi Poisson structures and noncommutative integrable systems.
		PhD thesis, Rutgers University-School of Graduate Studies, 2018.
		
		\bibitem{arthamonov20}
		S. Arthamonov, N. Ovenhouse and M. Shapiro.
		\newblock Noncommutative networks on a cylinder.
		\newblock arXiv: 2008.02889, 2020.
		
		
		\bibitem{open2018}
		A. Bolsinov, V. S. Matveev,
		E. Miranda and S. Tabachnikov. 
		\newblock Open
		problems, questions and challenges in finite dimensional
		integrable systems.
		\newblock
		\emph{ Phil. Trans. R.
			Soc. A},
		376(2018), 20170430.
		
		\bibitem{casati21}
		M. Casati and J. Wang.
		\newblock Recursion and Hamiltonian operators for integrable nonabelian difference equations.
		\newblock \emph{Nonlinearity}, 34 (2021), 205-236.
		
		\bibitem{gelfand05}
		I. Gelfand, S. Gelfand, V. Retakh and R. Wilson.
		\newblock{Quasi-determinants}.
		\newblock \emph{Adv. Math.}, 193 (2005), 56-141.
		
		\bibitem{gelfand95}
		I. Gelfand, D. Krob, A. Lascoux, B. Leclerc, V. Retakh and J. Thibon.
		\newblock {Non-commutative symmetric functions}.
		\newblock \emph{Adv. Math.}, 112 (1995), 218-348.
		
		
		\bibitem{gelfand91}
		I. Gelfand and V. Retakh.
		\newblock Determinants of matrices over non-commutative rings.
		\newblock \emph{Funct. Anal. Appl.}, 25 (1991), 91-102.
		
		\bibitem{gelfand97}
		I. Gelfand and V. Retakh.
		\newblock Quasideterminants. I.
		\newblock \emph{Selecta Math.}, 3 (1997), 517-546.
		
		\bibitem{gekhtman2009}
		M. Gekhtman, M. Shapiro and A. Vainshtein.
		\newblock Poisson geometry of directed networks in a disk.
		\newblock \emph{Selecta Math.}, 15 (2009), 61-103.
		
		\bibitem{gekhtman2012}
		M. Gekhtman,
		M. Shapiro and A. Vainshtein. 
		\newblock Poisson geometry of directed networks in an annulus. 
		\newblock \emph{J. Eur. Math. Soc. (JEMS)},
		14 (2012), no. 2, 541–570.
		
		\bibitem{gekhtman2016}
		M. Gekhtman,
		M. Shapiro, S. Tabachnikov and A. Vainshtein.
		\newblock Integrable cluster dynamics of directed networks and pentagram maps.
		\newblock \emph{Adv. Math.},	300 (2016), 390–450.
		
		\bibitem{gehktman11}
		M. Gekhtman and O. Korovnichenko.
		\newblock Matrix Weyl functions and non-abelian Coxeter-Toda lattices.
		\newblock \emph{Notions of Positivity and the Geometry of Polynomials (Trends in Mathematics)}, ed P Br\"and\'en et al, Springer, Berlin, pp 221-237, 2021.
		
		
		\bibitem{gilson07}
		C. Gilson, J. Nimmo and Y. Ohta.
		\newblock Quasideterminant solutions of a non-Abelian Hirota-Miwa equation.
		\newblock \emph{J. Phys. A}, 40, 12607, 2007.
		
		
		\bibitem{Y2011}
		M.~Glick.
		\newblock The pentagram map and {Y}-patterns.
		\newblock {\em Adv. Math.}, 227(2):1019 -- 1045, 2011.
		
		
		\bibitem{Y2016}
		M.~Glick and P.~Pylyavskyy.
		\newblock Y-meshes and generalized pentagram maps.
		\newblock {\em Proc. Lond. Math. Soc.}, 112(4):753--797, 2016.
		
		\bibitem{long2022}
		A.~Izosimov and B.~Khesin.
		\newblock Long-diagonal pentagram maps.
		\newblock {\em Bull. London Math. Soc.}, 
		55: 1314-1329, 2023.
		
		
		\bibitem{dented2016}
		B.~Khesin and F.~Soloviev.
		\newblock The geometry of dented pentagram maps.
		\newblock {\em J. Eur. Math. Soc.}, 18(1):147--179, 2016.
		
		
		\bibitem{higher2012}
		B.~Khesin and F.~Soloviev.
		\newblock Integrability of higher pentagram maps.
		\newblock {\em Math. Ann.}, 357(3):1005--1047, 2012.
		
		\bibitem{li23}
		S. Li.
		\newblock Matrix Orthogonal Polynomials, non-abelian Toda lattice and Bäcklund transformation.
		\newblock {\em Sci. China Math.}, DOI: 10.1007/s11425-022-2168-x.
		
		\bibitem{grassmann}
		G.~Marí~Beffa and R.~Felipe.
		\newblock The pentagram map on {G}rassmannians.
		\newblock {\em Ann. Inst. Fourier.}, 69(1):421--456, 2019.
		
				\bibitem{massuyeau14}
		G. Massuyeau and V. Turaev.
		\newblock Quasi-Poisson structures on representation space of surfaces.
		\newblock \emph{Int. Math. Res. Not.}, 2014(1): 1-64, 2012.
		
		
		\bibitem{Non}
		N.~Ovenhouse.
		\newblock Non-commutative integrability of the {G}rassmann pentagram map.
		\newblock {\em Adv. Math.}, 373:107309, 2020.
		
		\bibitem{Ovsi2010}
		V.~Ovsienko, R.~Schwartz, and S.~Tabachnikov.
		\newblock The {P}entagram map: A discrete integrable system.
		\newblock {\em Comm. Math. Phys.}, 299(2):409--446, 2010.
		
		\bibitem{retakh14} 
		V. Retakh.
		\newblock Non-commutative cross-ratios.
		\newblock \emph{J. Geom. Phys.}, 82 (2014), 13-17.
		
		\bibitem{retakh20}
		V. Retakh, V. Rubtsov and G. Sharygin.
		\newblock Non-commutative cross-ratio and Schwarz derivative.
		\newblock Integrable Systems and Algebraic Geometry, \emph{Lecture Notes London Math. Soc.}, 459 (2020), 499-528.
		
		\bibitem{1992The}
		R.~Schwartz.
		\newblock The pentagram map.
		\newblock {\em Exp. Math.}, 1(1):71--81, 1992.
		
		\bibitem{Soloviev2013Integrability}
		F.~Soloviev.
		\newblock Integrability of the pentagram map.
		\newblock {\em Duke Math. J.}, 162(15):2815--2853, 2013.
		
		\bibitem{van2008double}
		M. Van den Bergh,
		\newblock
		Double poisson algebras,
		\newblock
		\emph{Trans. Amer. Math. Soc.},
		360 (2008), no. 11,
		5711–5769.
		
		\bibitem{wang23}
		B. Wang and X. Chang.
		Pentagram maps on coupled polygons: integrability, geometry, limit points and orthogonality, submitted, 2023.
		
		\bibitem{op2022}
		B. Wang, X. K. Chang and X. L. Yue.
		\newblock
		A generalization of Laurent biorthogonal polynomials and related integrable lattices.
		\newblock
		\emph{J. Phys. A: Math. Theor.},
		55, 214002, 2022.
		
		\bibitem{young82}
		N. Young.
		\newblock Linear fractional transformations in rings and modules.
		\newblock {\em Linear Algebra and its Applications}, 56 (1984), 251-290.
		
		
	\end{thebibliography}
\end{document}